\RequirePackage[l2tabu, orthodox]{nag}
\documentclass[a4paper,11pt]{article}
\pdfoutput=1

\usepackage{fixltx2e}
\usepackage[T1]{fontenc}
\usepackage[utf8]{inputenc}
\usepackage{lmodern}
\usepackage{microtype}
\usepackage[margin=1in]{geometry}
\usepackage[UKenglish]{babel}
\usepackage[backend=biber,style=mynumeric]{biblatex}
\usepackage{csquotes}
\usepackage{latexsym}
\usepackage{mathrsfs}
\usepackage{amsmath}
\usepackage{amssymb}
\usepackage{amsthm}
\usepackage[all,warning]{onlyamsmath}
\usepackage{enumitem}
\usepackage{xparse}
\usepackage{mathtools}
\usepackage{thmtools}
\usepackage{thm-restate}
\usepackage{graphicx}
\usepackage{subfig}
\usepackage[hidelinks]{hyperref}

\theoremstyle{plain}
\newtheorem{theorem}{Theorem}
\newtheorem*{thm1}{Theorem 1}
\newtheorem*{thm2}{Theorem 2}
\newtheorem{lemma}[theorem]{Lemma}

\theoremstyle{definition}
\newtheorem{definition}[theorem]{Definition}

\theoremstyle{remark}
\newtheorem{remark}[theorem]{Remark}

\setlist[itemize]{label=--}
\setlist[enumerate]{label=(\arabic*),labelindent=\parindent,leftmargin=*}

\DeclarePairedDelimiter\braces{\{}{\}}
\NewDocumentCommand\set{O{}mg}{\ensuremath{\braces[#1]{#2\IfNoValueTF{#3}{}{\,:\,#3}}}}

\DeclareMathOperator{\dist}{dist}
\DeclareMathOperator{\dom}{dom}
\DeclareMathOperator{\vset}{set}
\DeclareMathOperator{\vmset}{multiset}

\newcommand{\N}{\mathbb{N}}
\newcommand{\Np}{\mathbb{N_+}}
\newcommand{\F}{\mathcal{F}}
\newcommand{\w}{\overline{v}}
\newcommand{\ww}{\overline{u}}
\newcommand{\nbsmsv}[1]{\hspace{4pt}\underline{\hspace{-1.5pt}\leftrightarrow\hspace{-1.5pt}}\hspace{1pt}_{#1}^\aSV\hspace{2pt}}

\newcommand{\aA}{\mathcal{A}}

\newcommand{\sA}{\mathbf{A}}

\newcommand{\aMV}{\mathcal{MV}}
\newcommand{\aSV}{\mathcal{SV}}

\newcommand{\sMV}{\mathbf{MV}}
\newcommand{\sSV}{\mathbf{SV}}

\newcommand{\sMB}{\mathbf{MB}}

\newcommand{\VVc}{\mathsf{VV_c}}
\newcommand{\VV}{\mathsf{VV}}
\newcommand{\MV}{\mathsf{MV}}
\newcommand{\SV}{\mathsf{SV}}
\newcommand{\VB}{\mathsf{VB}}
\newcommand{\MB}{\mathsf{MB}}
\newcommand{\SB}{\mathsf{SB}}

\newcommand{\MVl}{\MV(1)}
\newcommand{\SVl}{\SV(1)}

\newcommand{\cB}{\mathsf{B}}
\newcommand{\cW}{\mathsf{W}}
\newcommand{\cG}{\mathsf{G}}

\hypersetup{
  pdftitle={Ability to Count Is Worth Theta(Delta) Rounds},
  pdfauthor={Tuomo Lempiäinen}
}

\begin{document}

\title{\textbf{\boldmath Ability to Count Is Worth $\Theta(\Delta)$ Rounds}}
\author{Tuomo Lempiäinen \bigskip\\
  \small{Helsinki Institute for Information Technology HIIT \& Department of Computer Science,} \\
  \small{Aalto University, Finland}}
\date{}
\maketitle

\begin{abstract}
  Hella et al.\ (PODC 2012, Distributed Computing 2015) identified seven different models of distributed computing---one of which is the port-numbering model---and provided a complete classification of their computational power relative to each other. However, one of their simulation results involves an additive overhead of $2\Delta-2$ communication rounds, and it was not clear, if this is actually optimal. In this paper we give a positive answer: there is a matching linear-in-$\Delta$ lower bound. This closes the final gap in our understanding of the models, with respect to the number of communication rounds.
\end{abstract}

\section{Introduction}

This work studies the significance of being able to count the multiplicities of identical incoming messages in distributed algorithms. We compare two models: one, in which each node receives a \emph{set} of messages in each round, and another, in which each node receives a \emph{multiset} of messages in each round. It has been previously shown that the latter model can be simulated in the former model by allowing an additive overhead of linear in $\Delta$ communication rounds, where $\Delta$ is the maximum degree of the graph~\cite{hella15weak-models}. In this work we show that this is optimal: in some cases, linear in $\Delta$ extra rounds are strictly necessary.

\subsection{A Hierarchy of Weak Models}\label{sec:hierarchy}

The models that we study are weaker variants of the well-known \emph{port-numbering model}. \Textcite{hella15weak-models} defined a collection of seven models, one of which is the port-numbering model. We denote by $\VVc$ the class of all \emph{graph problems} that can be solved in this model. The following subclasses of $\VVc$ correspond to the weaker variants:
\begin{itemize}[noitemsep,leftmargin=4em]
  \item[$\VV$:] Input and output ports connected to the same neighbour do not necessarily have the same number.
  \item[$\MV$:] Input ports are not numbered; nodes receive a multiset of messages.
  \item[$\SV$:] Input ports are not numbered; nodes receive a set of messages.
  \item[$\VB$:] Output ports are not numbered; nodes broadcast the same message to all neighbours.
  \item[$\MB$:] Combination of $\MV$ and $\VB$.
  \item[$\SB$:] Combination of $\SV$ and $\VB$.
\end{itemize}

There are some trivial containment relations between the classes, such as $\SV \subseteq \MV \subseteq \VV \subseteq \VVc$. The trivial relations are depicted in Figure~\ref{fig:hierarchy1}. However, some classes, such as $\VB$ and $\SV$, are seemingly orthogonal. Somewhat surprisingly, \textcite{hella15weak-models} were able to show that the classes form a linear order:
\[
  \SB \subsetneq \MB = \VB \subsetneq \SV = \MV = \VV \subsetneq \VVc.
\]
For each class, we can also define the subclass of problems solvable in constant time independent on the size of the input graph. The same containment relations hold for the constant-time versions of the classes. The relations are depicted in Figure~\ref{fig:hierarchy2}.

The equalities between classes are proved by showing that algorithms corresponding to a seemingly more powerful class can be simulated by algorithms corresponding to a seemingly weaker class. In the case of $\SV = \MV$, there is an overhead involved, whereas the rest of the simulation results do not increase the running time.

\begin{figure}
  \centering
  \subfloat[]{\label{fig:hierarchy1}\includegraphics[page=5]{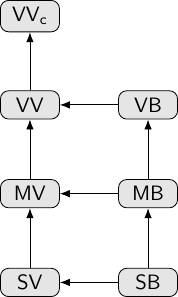}}
  \hspace{4em}
  \subfloat[]{\label{fig:hierarchy2}\includegraphics[page=6]{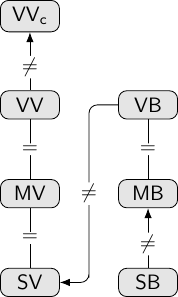}}
  \caption{\protect\subref{fig:hierarchy1} Trivial containment relations between the problem classes. \protect\subref{fig:hierarchy2} The linear order obtained by \textcite{hella15weak-models}.}
  \label{fig:hierarchy}
\end{figure}

\subsection{Classes \texorpdfstring{$\SV$}{SV} and \texorpdfstring{$\MV$}{MV}}

In this work we study further the relationship between the models that are related to the classes $\SV$ and $\MV$. Neither of the models features incoming port numbers. The only difference is that in the case of $\MV$, algorithms are able to count the number of neighbours that sent any particular message, while in $\SV$ this is not possible. For now we will use informally the terms \emph{$\SV$-algorithm} and \emph{$\MV$-algorithm}; a more formal definition will follow in Section~\ref{sec:prelim}.

\Textcite{hella15weak-models} proved that any $\MV$-algorithm can be simulated by an $\SV$-algorithm, given that the simulating algorithm is allowed to use $2\Delta-2$ extra communication rounds. The basic idea is that when nodes gather all available information from their radius-($2\Delta-2$) neighbourhood, the outgoing port numbers necessarily break symmetry. Any neighbours $u$ and $w$ of a node~$v$ either have different outgoing port numbers towards $v$ or are in a different internal state. This symmetry breaking information can then be used during the simulation to receive a distinct message from each neighbour.

\subsection{Contributions}

This work gives tight lower bounds for simulating $\MV$-algorithms by $\SV$-algorithms. We will prove two theorems. The first theorem is about a so-called simulation problem, that is, breaking symmetry between incoming messages. It is intended to be an exact counterpart to the upper bound result given by the simulation algorithm of \textcite{hella15weak-models}.

\begin{theorem}\label{thm:simulation}
  For each $\Delta \ge 2$ there is a port-numbered graph of maximum degree~$\Delta$ with nodes $v,u,w$, such that when executing any $\SV$-algorithm in the graph, node~$v$ receives identical messages from its neighbours $u$ and $w$ in rounds $1,2,\dotsc,2\Delta-2$.
\end{theorem}

Our second theorem gives a \emph{graph problem} that separates $\MV$-algorithms from $\SV$-algorithms with respect to running time as a function of the maximum degree~$\Delta$.

\begin{theorem}\label{thm:parity}
  There is a graph problem that can be solved in one communication round by an $\MV$-algorithm, but that requires at least $\Delta$ rounds for all $\Delta \ge 2$, when solved by an $\SV$-algorithm.
\end{theorem}

Our results are based on constructing a family of graphs with an intricate port numbering of certain kind. We start by proving Theorem~\ref{thm:simulation} in Section~\ref{sec:simlb}, and then we adapt the same construction to prove Theorem~\ref{thm:parity} in Section~\ref{sec:graphlb}.

\subsection{Motivation and Related Work}

The port-numbering model, or $\VVc$, can be thought to model wired networks, whereas model~$\SB$ corresponds to fully wireless systems. Other models in the hierarchy are intermediate steps between the two extreme cases.

Models similar to $\MV$ have been studied previously under various names:
output port awareness~\cite{boldi96symmetry},
wireless in input~\cite{boldi97computing},
mailbox~\cite{boldi97computing},
port-to-mailbox~\cite{yamashita99leader} and
port-\`a-bo\^\i te~\cite{chalopin06phd}.
However, most of the previous research does not give general results about graph problems, but instead focuses on individual problems or makes different assumptions about the model. To the best of our knowledge, model~$\SV$ has not been studied before the work of \textcite{hella15weak-models}.

\Textcite{emek13stone} have considered networks of nodes with very limited computation and communication capabilities---in particular, the nodes can count identical messages only up to some predetermined number. They argue that models like that will be crucial when applying distributed computing to networks of biological cells.

Our models have analogies also in graph exploration.
Models $\SV$ and $\MV$ correspond to the case where an agent does not know from which edge it arrived to a node. This is true for \emph{traversal sequences}~\cite{aleliunas79random}, as opposed to \emph{exploration sequences}~\cite{koucky02universal}. If we have several agents exploring a graph, the question of whether they can count the number of identical agents in a node becomes interesting. Our lower bounds indicate that, with appropriate definitions, this ability causes a difference of linear in $\Delta$ steps in certain traversal sequences.

\Textcite{hella15weak-models} identified a connection between the seven models of computation and certain variants of modal logic, in the spirit of descriptive complexity theory. In certain classes of structures, \emph{graded multimodal logic} corresponds to $\MV$ and \emph{multimodal logic} corresponds to $\SV$. Thus our lower-bound result can be recast in terms of modal formulas: when given a formula~$\phi$ of graded multimodal logic, we can find an equivalent formula~$\psi$ of multimodal logic, but in general, the modal depth $\operatorname{md}(\psi)$ of $\psi$ has to be at least $\operatorname{md}(\phi)+\Delta-1$. For details on modal logic, see \textcite{blackburn01modal} or \textcite{blackburn07handbook}.

\section{Preliminaries}\label{sec:prelim}

In this section we define the models of computation and the problems we study, as well as introduce tools that will be needed in order to prove our results.

\subsection{Distributed Algorithms}\label{sec:distalg}

We define distributed algorithms as state machines. They are executed in a graph such that each node of the graph is a copy of the same state machine. Nodes can communicate with adjacent nodes. In this work, we consider only deterministic state machines and synchronous communication in anonymous networks.

In the beginning of execution, each state machine is initialised based on the degree of the node and a possible local input given to it. Then, in each communication round, each state machine performs three operations:
\begin{enumerate}[noitemsep]
  \item sends a message to each neighbour,
  \item receives a message from each neighbour,
  \item moves to a new state based on the current state and the received messages.
\end{enumerate}
If the new state is a special stopping state, the machine halts. The local output of the node is its state after halting. Next, we will define distributed systems more formally.

\subsubsection{Inputs and Port Numberings}

Consider a graph~$G=(V,E)$. An \emph{input} for $G$ is a function $f \colon V \to X$, where $X$ is a finite set such that $\emptyset \in X$. For each $v \in V$, the value $f(v)$ is called the \emph{local input} of $v$.

A \emph{port} of $G$ is a pair $(v,i)$, where $v \in V$ is a node and $i \in [\deg(v)]$ is the number of the port. Let $P(G)$ be the set of all ports of $G$. A \emph{port numbering} of $G$ is a bijection $p \colon P(G) \to P(G)$ such that
\[
  p(v,i) = (u,j) \text{ for some $i$ and $j$} \quad \text{if and only if} \quad \set{v,u} \in E.
\]
Intuitively, if $p(v,i) = (u,j)$, then $(v,i)$ is an output port of node~$v$ that is connected to an input port $(u,j)$ of node~$u$.

When analysing lower-bound constructions, we will find the following generalisation of port numbers useful. Let $N$ be an arbitrary set. Assume that for each $v \in V$, $I_v \subseteq N$ and $O_v \subseteq N$ are subsets of size~$\deg(v)$. Now, a \emph{generalised input port} is a pair $(v,i)$, where $v \in V$ and $i \in I_v$, and a \emph{generalised output port} is a pair $(v,o)$, where $v \in V$ and $o \in O_v$. A \emph{generalised port numbering}~$p$ is then a bijection that maps each output port to an input port of an adjacent node.

\subsubsection{State Machines}\label{sec:statemach}

For each positive integer~$\Delta$, denote by $\F(\Delta)$ the class of all simple undirected graphs of maximum degree at most~$\Delta$. Let $X \ni \emptyset$ be a finite set of local inputs. A \emph{distributed state machine} for $(\F(\Delta),X)$ is a tuple $\aA = (Y,Z,\sigma_0,M,\mu,\sigma)$, where
\begin{itemize}[noitemsep]
  \item $Y$ is a set of states,
  \item $Z \subseteq Y$ is a finite set of stopping states,
  \item $\sigma_0 \colon \set{0,1,\dotsc,\Delta} \times X \to Y$ is a function that defines the initial state,
  \item $M$ is a set of messages such that $\epsilon \in M$,
  \item $\mu \colon Y \times [\Delta] \to M$ is a function that constructs the outgoing messages, such that $\mu(z,i) = \epsilon$ for all $z \in Z$ and $i \in [\Delta]$,
  \item $\sigma \colon Y \times M^\Delta \to Y$ is a function that defines the state transitions, such that $\sigma(z,\overline{m}) = z$ for all $z \in Z$ and $\overline{m} \in M^\Delta$.
\end{itemize}
The special symbol $\epsilon \in M$ indicates ``no message'' and $\emptyset$ indicates ``no input''.

\subsubsection{Executions}

Let $G = (V,E) \in \F(\Delta)$ be a graph, let $p$ be a port numbering of $G$, let $f \colon V \to X$ be an input for $G$, and let $\aA$ be a distributed state machine for $(\F(\Delta),X)$. Then we can define the \emph{execution} of $\aA$ in $(G,f,p)$ as follows.

The state of the system in round~$r \in \N$ is represented as a function $x_r \colon V \to Y$, where $x_r(v)$ is the \emph{state} of node~$v$ in round~$r$. To initialise the nodes, set
\[
  x_0(v) = \sigma_0(\deg(v),f(v)) \quad \text{for each } v \in V.
\]
Then, assume that $x_r$ is defined for some $r \in \N$. Let $(u,j) \in P(G)$ and $(v,i) = p(u,j)$. Now, node~$v$ receives the message
\[
  a_{r+1}(v,i) = \mu(x_r(u),j)
\]
from its port $(v,i)$ in round~$r+1$. For each $v \in V$, we define a vector of length~$\Delta$ consisting of messages received by node~$v$ in round~$r+1$ and the symbol~$\epsilon$:
\[
  \overline{a}_{r+1}(v) = (a_{r+1}(v,1),a_{r+1}(v,2),\dotsc,a_{r+1}(v,\deg(v)),\epsilon,\epsilon,\dotsc,\epsilon),
\]
where the padding with the special symbol~$\epsilon$ is to simplify our notation so that $\overline{a}_{r+1}(v) \in M^\Delta$. Now we can define the new state of each node~$v \in V$ as follows:
\[
  x_{r+1}(v) = \sigma(x_r(v),\overline{a}_{r+1}(v)).
\]
Let $t \in \N$. If $x_t(v) \in Z$ for all $v \in V$, we say that $\aA$ \emph{stops in time~t} in $(G,f,p)$. The \emph{running time} of $\aA$ in $(G,f,p)$ is the smallest~$t$ for which this holds. If $\aA$ stops in time~$t$ in $(G,f,p)$, the \emph{output} of $\aA$ in $(G,f,p)$ is $x_t \colon V \to Y$. For each $v \in V$, the \emph{local output} of $v$ is $x_t(v)$.

We define the \emph{execution} of $\aA$ in $(G,p)$ to be the execution of $\aA$ in $(G,f,p)$, where $f$ is the unique function $f \colon V \to \set{\emptyset}$.

\subsubsection{Algorithm Classes}

So far, we have defined only a single model of computation. However, our aim in this work is to investigate the relationships between two variants of the model. To this end, we will now introduce two different restrictions to the definition of a state machine.

Given a vector $\overline{a} = (a_1,a_2,\dotsc,a_\Delta) \in M^\Delta$, define
\begin{align*}
  \vset(\overline{a}) &= \set{a_1,a_2,\dotsc,a_\Delta}, \\
  \vmset(\overline{a}) &= \set{(m,n)}{m \in M, n = |\set{i \in [\Delta]}{m = a_i}|}.
\end{align*}
That is, $\vset(\overline{a})$ discards the ordering and multiplicities of the elements of $\overline{a}$, while $\vmset(\overline{a})$ discards only the ordering.

Now we can define classes $\aSV$ and $\aMV$ of state machines. Class $\aSV$ consists of all distributed state machines $\aA = (Y,Z,\sigma_0,M,\mu,\sigma)$ such that
\[
  \vset(\overline{a}) = \vset(\overline{b}) \quad\text{implies}\quad \sigma(y,\overline{a}) = \sigma(y,\overline{b}) \quad\text{for all}\quad y \in Y.
\]
Similarly, class $\aMV$ consists of all distributed state machines $\aA = (Y,Z,\sigma_0,M,\mu,\sigma)$ such that
\[
  \vmset(\overline{a}) = \vmset(\overline{b}) \quad\text{implies}\quad \sigma(y,\overline{a}) = \sigma(y,\overline{b}) \quad\text{for all}\quad y \in Y.
\]

The idea here is that for state machines in $\aMV$, the state transitions are invariant with respect to the order of incoming messages; in practice, nodes receive the messages in a multiset. In $\aSV$, nodes receive the messages in a set, which means that the state transitions are invariant with respect to both the order and multiplicities of incoming messages.

We will later find useful the following definitions for infinite sequences of state machines, where $\Delta$ will be used as an upper bound for the maximum degree of graphs:
\begin{align*}
  \sMV &= \set{(\aA_1,\aA_2,\dots)}{\aA_\Delta \in \aMV \text{ for all } \Delta}, \\
  \sSV &= \set{(\aA_1,\aA_2,\dots)}{\aA_\Delta \in \aSV \text{ for all } \Delta}.
\end{align*}

From now on, both distributed state machines~$\aA$ and sequences of distributed state machines~$\sA$ will be referred to as \emph{algorithms}. The precise meaning should be clear from the notation.

\subsection{Graph Problems}\label{sec:graphprob}

Let $X$ and $Y$ be finite nonempty sets. A \emph{graph problem} is a function~$\Pi_{X,Y}$ that maps each undirected simple graph~$G=(V,E)$ and each input~$f \colon V \to X$ to a set~$\Pi_{X,Y}(G,f)$ of solutions. Each \emph{solution} $S \in \Pi_{X,Y}(G,f)$ is a function~$S \colon V \to Y$. We handle problems without local input by setting $X = \set{\emptyset}$. One can see that our definition covers a large selection of typical distributed computing problems, such as those where the task is to find a subset or colouring of vertices.

Let $\Pi_{X,Y}$ be a graph problem, $T \colon \N \times \N \to \N$ a function and $\sA = (\aA_1,\aA_2,\dotsc)$ a sequence such that each $\aA_\Delta$ is a distributed state machine for $(\F(\Delta),X)$. We define that \emph{$\sA$ solves $\Pi_{X,Y}$ in time~$T$} if the following conditions hold for all $\Delta \in \N$, all finite graphs~$G=(V,E) \in \F(\Delta)$, all inputs $f \colon V \to X$ and all port numberings~$p$ of $G$:
\begin{enumerate}[noitemsep]
  \item $\aA_\Delta$ stops in time $T(\Delta,|V|)$ in $(G,f,p)$.
  \item The output of $\aA_\Delta$ in $(G,f,p)$ is in $\Pi_{X,Y}(G,f)$.
\end{enumerate}

If there exists a function $T \colon \N \times \N \to \N$ such that $\sA$ solves $\Pi_{X,Y}$ in time~$T$, we say that \emph{$\sA$ solves $\Pi_{X,Y}$} or that \emph{$\sA$ is an algorithm for $\Pi_{X,Y}$}. If the value~$T(\Delta,n)$ does not depend on $n$, that is, if we have $T(\Delta,n) = T'(\Delta)$ for some function $T' \colon \N \to \N$, we say that \emph{$\sA$ solves $\Pi_{X,Y}$ in constant time} or that \emph{$\sA$ is a local algorithm for $\Pi_{X,Y}$}.

\begin{remark}
  Local inputs do not add anything essential to our work. Since the set~$X$ of possible input values is uniformly finite, the information given by an input $f \colon V \to X$ could be encoded as topological information in the graph. However, the use of local inputs will make our life easier, when we construct problem instances in Section~\ref{sec:graphlb}.
\end{remark}

\subsubsection{Problem Classes}

Now we are ready to define complexity classes based on our different notions of algorithms. The two classes studied in this work are as follows:
\begin{itemize}[noitemsep]
  \item $\MV$ consists of problems~$\Pi$ such that there is an algorithm $\sA \in \sMV$ that solves $\Pi$.
  \item $\SV$ consists of problems~$\Pi$ such that there is an algorithm $\sA \in \sSV$ that solves $\Pi$.
\end{itemize}
For both classes, we can also define their constant-time variants:
\begin{itemize}[noitemsep]
  \item $\MVl$ consists of problems~$\Pi$ such that there is $\sA \in \sMV$ that solves $\Pi$ in constant time.
  \item $\SVl$ consists of problems~$\Pi$ such that there is $\sA \in \sSV$ that solves $\Pi$ in constant time.
\end{itemize}

Observe that it follows trivially from the definitions of the algorithm classes that $\SV \subseteq \MV$ and $\SVl \subseteq \MVl$. It was shown by \textcite{hella15weak-models} that we actually have $\SV = \MV$ and $\SVl = \MVl$.

\subsection{Bisimulation}\label{sec:bisim}

In this section we introduce a tool that we will need when proving lower-bound results in Sections~\ref{sec:simlb} and \ref{sec:graphlb}. The tool in question is bisimulation, and in particular, its finite approximation, which we call $r$-bisimulation. Simply put, a bisimulation is a relation between two structures such that related elements have identical local information and equivalent relations to other elements. For more details on bisimulation in general, see \textcite{blackburn01modal} or \textcite{blackburn07handbook}.

\textcite{hella15weak-models} demonstrated the use of bisimulation in distributed computing by establishing a connection between the weak models mentioned in Section~\ref{sec:hierarchy} and certain variants of modal logic. Here we take a considerably simpler approach and show directly that bisimilarity implies indistinguishability by distributed algorithms.

The general concept of bisimulation can be adapted to take into account the different amounts of information that is available to algorithms in each model. We will need only one variant in this work, the one corresponding to the class~$\aSV$.

\begin{definition}\label{def:nbisimsv}
  Let $G = (V,E)$ and $G' = (V',E')$ be graphs, let $f$ and $f'$ be inputs for $G$ and $G'$, respectively, and let $p$ and $p'$ be generalised port numberings of $G$ and $G'$, respectively. We define $r$-$\aSV$-bisimilarity recursively. As a base case, we say that nodes $v \in V$ and $v' \in V'$ are \emph{$0$-$\aSV$-bisimilar} if $\deg_G(v) = \deg_{G'}(v')$ and $ f(v) = f'(v')$. For $r \in \Np$, we say that $v \in V$ and $v' \in V'$ are \emph{$r$-$\aSV$-bisimilar} if the following conditions hold:
  \begin{enumerate}[label=(B\arabic*)]
    \item Nodes $v$ and $v'$ are $0$-$\aSV$-bisimilar.
    \item If $\set{v,w} \in E$, then there is $w' \in V'$ with $\set{v',w'} \in E'$ such that $w$ and $w'$ are $(r-1)$-$\aSV$-bisimilar, and $p(w,a) = (v,b)$ and $p'(w',a) = (v',c)$ hold for some $a$, $b$, $c$.
    \item If $\set{v',w'} \in E'$, then there is $w \in V$ with $\set{v,w} \in E$ such that $w$ and $w'$ are $(r-1)$-$\aSV$-bisimilar, and $p(w,a) = (v,b)$ and $p'(w',a) = (v',c)$ hold for some $a$, $b$, $c$.
  \end{enumerate}
  If $v \in V$ and $v' \in V'$ are \emph{$r$-$\aSV$-bisimilar}, we write $(G,f,v,p) \nbsmsv{r} (G',f',v',p')$---or simply $v \nbsmsv{r} v'$, if the graphs, inputs and generalised port numberings are clear from the context.
\end{definition}

It is clear from the definition that if $(G,f,v,p) \nbsmsv{r} (G',f',v',p')$ holds for some $r$, then $(G,f,v,p) \nbsmsv{t} (G',f',v',p')$ holds for all $t = 0,1,\dotsc,r$. As the following lemma shows, $r$-bisimilarity entails indistinguishability by distributed algorithms up to running time~$r$.

\begin{restatable}{lemma}{lemnbsmsvstate}\label{lem:nbsmsvstate}
  Let $G = (V,E)$ and $G' = (V',E')$ be graphs, let $f$ and $f'$ be inputs for $G$ and $G'$, respectively, and let $p$ and $p'$ be port numberings of $G$ and $G'$, respectively. If $(G,f,v,p) \nbsmsv{r} (G',f',v',p')$ for some $r \in \N$, $v \in V$ and $v' \in V'$, then for all algorithms $\aA \in \aSV$ we have $x_t(v) = x'_t(v')$ for all $t = 0,1,\dotsc,r$, that is, the states of $v$ and $v'$ are identical in rounds~$0,1,\dotsc,r$.
\end{restatable}

\begin{proof}
  We prove the claim by induction on $r$. Let $\aA \in \aSV$ be an arbitrary algorithm.
  The base case $r = 0$ is clear: since $v \nbsmsv{0} v'$, we have
  \[
    x_0(v) = \sigma_0(\deg(v),f(v)) = \sigma_0(\deg(v'),f'(v')) = x'_0(v').
  \]

  Suppose then that the claim holds for $r = s$ and that $v \nbsmsv{s+1} v'$. We obtain immediately by the inductive hypothesis that $x_t(v) = x'_t(v')$ for all $t=0,1,\dotsc,s$. Conditions (B2) and (B3) of Definition~\ref{def:nbisimsv} guarantee that for each neighbour~$u$ of $v$ there is a neighbour~$u'$ of $v'$, and vice versa, such that $u \nbsmsv{s} u'$, and additionally, $p(u,j) = (v,i)$ and $p'(u',j) = (v',i')$ for some $j,i,i'$. For each such pair of neighbours, the inductive hypothesis implies that $x_s(u) = x'_s(u')$. We have now $\mu(x_s(u),j) = \mu(x'_s(u'),j)$ and thus $a_{s+1}(v,i) = a'_{s+1}(v',i')$. That is, for each message $a_{s+1}(v,k)$ in the vector $\overline{a}_{s+1}(v)$ there is an identical message $a'_{s+1}(v',k')$ in $\overline{a}'_{s+1}(v')$, and vice versa. Additionally, as $\deg(v) = \deg(v')$, the special symbol~$\epsilon$ is either in both of the vectors or in neither of them. It follows that $\vset(\overline{a}_{s+1}(v)) = \vset(\overline{a}'_{s+1}(v'))$. Since $\aA \in \aSV$, we have
  \[
    x_{s+1}(v) = \sigma(x_s(v),\overline{a}_{s+1}(v)) = \sigma(x'_s(v'),\overline{a}'_{s+1}(v')) = x'_{s+1}(v').
  \]
  Now $x_t(v) = x'_t(v')$ for all $t=0,1,\dotsc,s+1$, and hence we have shown that the claim holds for $r = s+1$.
\end{proof}

It is quite straightforward to show by induction that $r$-$\aSV$-bisimilarity is an equivalence relation. Since we will only need transitivity in this work, the following lemma suffices.

\begin{restatable}{lemma}{lemnbsmeqrel}\label{lem:nbsmeqrel}
  The $r$-$\aSV$-bisimilarity relation $\nbsmsv{r}$ is transitive in the class of quadruples $(G,f,v,p)$, where $G = (V,E)$ is a graph, $f$ is an input for $G$, $p$ is a generalised port numbering of $G$ and $v \in V$.
\end{restatable}

\begin{proof}
  We proceed by induction on $r$. The base case $r = 0$ is clear: if $(G,f,v,p) \nbsmsv{0} (G',f',v',p')$ and $(G',f',v',p') \nbsmsv{0} (G'',f'',v'',p'')$, then $(G,f,v,p) \nbsmsv{0} (G'',f'',v'',p'')$. Suppose then that relation~$\nbsmsv{r}$ is transitive for $r = s$ and that we have $(G,f,v,p) \nbsmsv{s+1} (G',f',v',p')$ and $(G',f',v',p') \nbsmsv{s+1} (G'',f'',v'',p'')$. Condition (B1) for $v$ and $v''$ is equivalent to the base case. If $\set{v,u} \in E$, condition~(B2) for $v$ and $v'$ implies that there is $u' \in V'$ with $\set{v',u'} \in E'$ such that $u \nbsmsv{s} u'$, and additionally, $p(u,j) = (v,i)$ and $p'(u',j) = (v',i')$ for some $j,i,i'$. Then, condition~(B2) for $v'$ and $v''$ implies that there is $u'' \in V''$ with $\set{v'',u''} \in E''$ such that $u' \nbsmsv{s} u''$ and $p''(u'',j) = (v'',i'')$ for some $i''$. By the inductive hypothesis, we have $u \nbsmsv{s} u''$, and thus $v$ and $v''$ satisfy condition~(B2). The case of the reverse condition~(B3) is very similar. We obtain $(G,f,v,p) \nbsmsv{s+1} (G'',f'',v'',p'')$, which shows that $\nbsmsv{r}$ is transitive for $r = s+1$.
\end{proof}

Finally, when given a generalised port numbering and a bisimilarity result, we need to be able to introduce an ordinary port numbering in order to actually apply the result to distributed algorithms. The following lemma shows that we can do this.

\begin{restatable}{lemma}{lempnchangebsm}\label{lem:pnchangebsm}
  Let $G = (V,E)$ and $G' = (V',E')$ be graphs, let $f$ and $f'$ be inputs for $G$ and $G'$, respectively, and let $p$ and $p'$ be generalised port numberings of $G$ and $G'$, respectively, with port numbers taken from a set~$N$. Suppose that $q$ and $q'$ are port numberings of $G$ and $G'$, respectively, such that $p(v,i) = (u,j)$ implies $q(v,g(i)) = (u,g(j))$ and $p'(v,i) = (u,j)$ implies $q'(v,g(i)) = (u,g(j))$ for some function $g \colon N \to \Np$. Then $(G,f,v,p) \nbsmsv{r} (G',f',v',p')$ implies $(G,f,v,q) \nbsmsv{r} (G',f',v',q')$ for all $v \in V$ and $v' \in V'$.
\end{restatable}

\begin{proof}
  We prove the claim by induction on $r$. The base case $r = 0$ is clear, since it does not depend on (generalised) port numbers. Suppose then that the claim holds for $r = s$ and $(G,f,v,p) \nbsmsv{s+1} (G',f',v',p')$. Condition~(B1) is equivalent to the base case. If $\set{v,u} \in E$, then by condition~(B2) there is $u' \in V'$ with $\set{v',u'}$ such that $(G,f,u,p) \nbsmsv{s} (G',f',u',p')$, and $p(u,j) = (v,i)$ and $p'(u',j) = (v',i')$ hold for some $j,i,i'$. From the inductive hypothesis we get $(G,f,u,q) \nbsmsv{s} (G',f',u',q')$, and by assumption, $q(u,g(j)) = (v,g(i))$ and $q'(u',g(j)) = (v',g(i'))$. Hence condition~(B2) holds also with respect to $q$ and $q'$. The case~(B3) is similar. This shows that $(G,f,v,q) \nbsmsv{s+1} (G',f',v',q')$ and thus the claim holds for $r = s+1$.
\end{proof}

\section{A Lower Bound for the Simulation Overhead}\label{sec:simlb}

Let us begin by restating the result that we will prove in this section.

\begin{thm1}
  For each $\Delta \ge 2$ there is a graph~$G = (V,E) \in \F(\Delta)$, a port numbering~$p$ of $G$ and nodes $v,u,w \in V$ such that when executing any algorithm~$\aA \in \aSV$ in $(G,p)$, node~$v$ receives identical messages from its neighbours $u$ and $w$ in rounds $1,2,\dotsc,2\Delta-2$.
\end{thm1}

To prove Theorem~\ref{thm:simulation}, we define for each $d = 2,3,\dotsc$ a graph~$G_d = (V_d,E_d)$ of maximum degree~$d$. The graph itself is just a rooted tree, but it gives rise to a port numbering with certain properties. The set~$V_d$ of nodes consists of sequences of pairs~$(i,j)$, where $i,j \in \set{0,1,\dotsc,d}$ will serve as a basis for port numbers, as we will see later. The sequence can be thought as a path leading from the root to the node itself. Our fundamental idea is that we construct the graph one level of nodes at a time, starting from the root, and assign generalised port numbers to each edge of a node by choosing the smallest numbers that have not yet been taken. The choice depends slightly on whether the level in question is even or odd.

We define the set~$V_d$ of nodes recursively as follows:
\begin{enumerate}[label=(G\arabic*)]
  \item $\emptyset \in V_d.$
  \item $((1,0)), ((2,1)), ((3,2)), ((4,3)),\dotsc,((d,d-1)) \in V_d.$
  \item If $(a_1,a_2,\dotsc,a_i) \in V_d$, where $i$ is odd and $i < 2d$, then $(a_1,a_2,\dotsc,a_{i+1}^j) \in V_d$ for all $j = 1,2,\dotsc,d-1$, where $a_{i+1}^j = (c_1^j, c_2^j)$ is defined as follows. Let $(b_1,b_2) = a_i$ and $b_2^+ = 1$ if $b_2 = 0$, $b_2^+ = b_2$ otherwise. Define
    \begin{align*}
      c_1^j &= \min(\set{1,2,\dotsc,d} \setminus \set{b_2^+,c_1^1,c_1^2,\dotsc,c_1^{j-1}}),\\
      c_2^j &= \min(\set{1,2,\dotsc,d} \setminus \set{b_1,c_2^1,c_2^2,\dotsc,c_2^{j-1}}).
    \end{align*}
  \item If $(a_1,a_2,\dotsc,a_i) \in V_d$, where $i$ is even and $0 < i < 2d$, then $(a_1,a_2,\dotsc,a_{i+1}^j) \in V_d$ for all $j = 1,2,\dotsc,d-1$, where $a_{i+1}^j = (c_1^j, c_2^j)$ is defined as follows. Let $(b_1,b_2) = a_i$. Define
    \begin{align*}
      c_1^j &= \min(\set{1,2,\dotsc,d} \setminus \set{b_2,c_1^1,c_1^2,\dotsc,c_1^{j-1}}),\\
      c_2^j &= \min(\set{0,1,\dotsc,d-1} \setminus \set{b_1,c_2^1,c_2^2,\dotsc,c_2^{j-1}}).
    \end{align*}
\end{enumerate}
The set $E_d$ of edges consists of all pairs $\set{v,u}$, where $v = (a_1,a_2,\dotsc,a_i) \in V_d$ and $u = (a_1,a_2,\dotsc,a_i,a_{i+1}) \in V_d$ for some $i \in \set{0,1,\dotsc}$. See Figure~\ref{fig:graphg} for an illustration of the radius-3 neighbourhood of node $\emptyset$ of $G_5$.

\begin{figure}[p]
  \captionsetup{skip=2cm}
  \centering
  \includegraphics[page=1,scale=0.85]{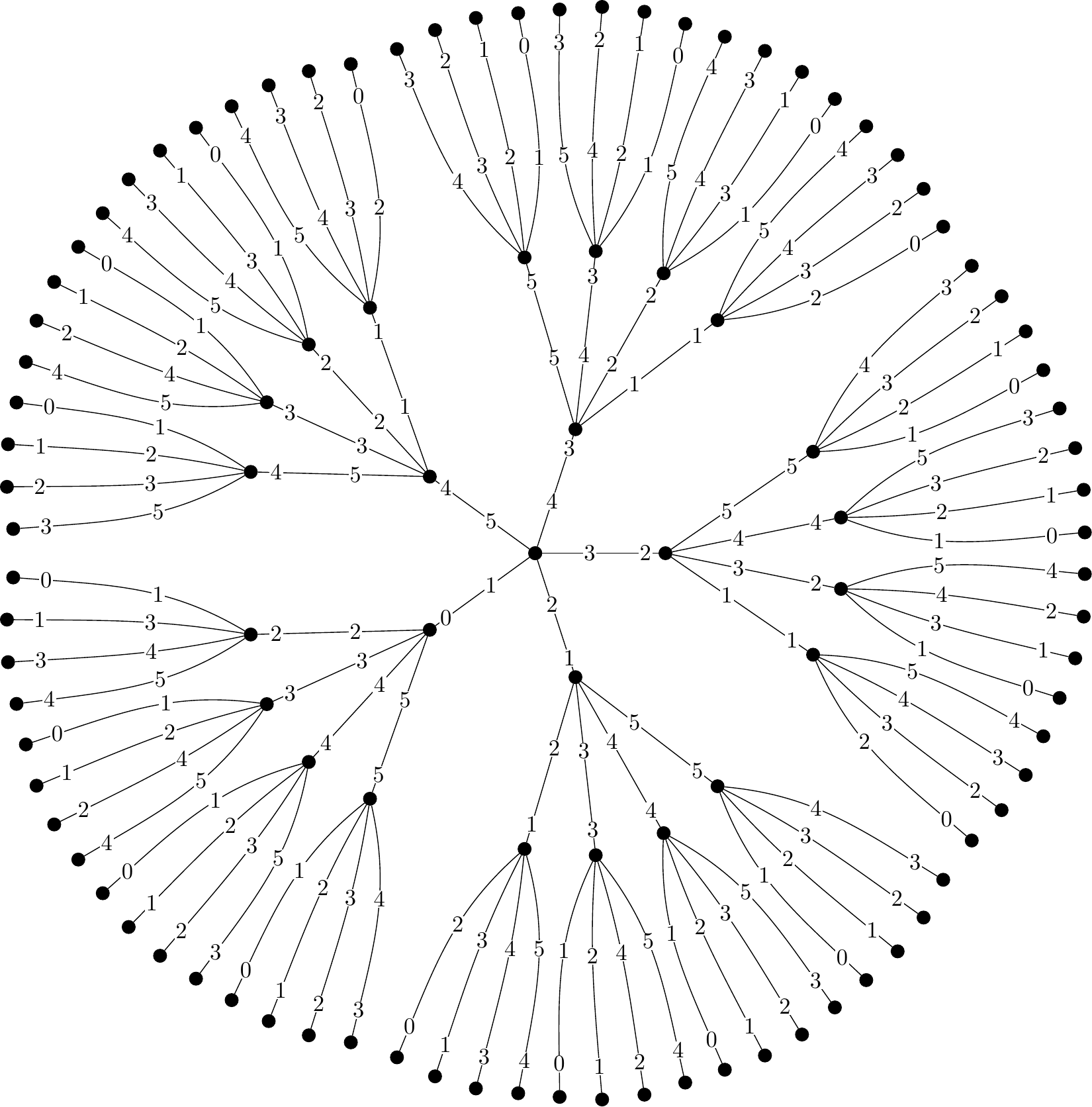}
  \caption{A part of the graph~$G_5$. The node in the centre is node~$\emptyset$. The numbers pictured are outgoing port numbers.}
  \label{fig:graphg}
\end{figure}

Consider nodes $v = (a_1,a_2,\dotsc,a_i)$ and $u = (a_1,a_2,\allowbreak\dotsc,a_{i+1})$, where $a_{i+1} = (b_1,b_2)$. The values $b_1$ and $b_2$ serve as generalised port numbers for the edge~$\set{v,u}$. We define $p_d(v,b_1) = (u,b_2)$ and $p_d(u,b_2) = (v,b_1)$. The incoming port numbers will be irrelevant in this proof, since we only consider algorithms in the classes $\aSV$ and $\aMV$. Thus, we will mostly use the notation $\pi_d(v,u) = b_1$ and $\pi_d(u,v) = b_2$ to denote the outgoing port numbers.

If $v = (a_1,a_2,\dotsc,a_i)$ and $u = (a_1,a_2,\dotsc,a_{i+1})$, we say that node~$v$ is the \emph{parent} of node~$u$ and that $u$ is a \emph{child} of $v$. We say that the node~$v$ is \emph{even} if $i$ is even and \emph{odd} if $i$ is odd. If $a_i = (b_1,b_2)$, we call $(b_1,b_2)$ the \emph{type} of node~$v$.

A \emph{walk} is a sequence $\w = (v_0,v_1,\dotsc,v_k)$ of nodes such that $\set{v_i,v_{i+1}} \in E_d$ for all $i = 0,1,\dotsc,k-1$. A pair~$(\w_1,\w_2)$ of walks, where $\w_i = (v_0^i,v_1^i,\dotsc,v_k^i)$ for all $i = 1,2$, and $k \le 2d-3$, is called a \emph{pair of compatible walks (PCW) of length~$k$ in $G_d$} if the following two conditions hold:
\begin{enumerate}[label=(W\arabic*)]
  \item $v_0^1 = ((1,0))$ and $v_0^2 = ((2,1))$.
  \item $\pi_d(v_j^1,v_{j-1}^1) = \pi_d(v_j^2,v_{j-1}^2)$ for all $j = 1,2,\dotsc,k$.
\end{enumerate}
If we additionally have the following for a PCW, it is called a \emph{pair of separating walks (PSW)}:
\begin{enumerate}[label=(W\arabic*),resume]
  \item There is $v_{k+1}^1 \in V_d$ with $\set{v_k^1,v_{k+1}^1} \in E_d$ such that there is no $v_{k+1}^2 \in V_d$ for which $\set{v_k^2,v_{k+1}^2} \in E_d$ and $\pi_d(v_{k+1}^1,v_k^1) = \pi_d(v_{k+1}^2,v_k^2)$.
\end{enumerate}
We say that a pair of separating walks of length~$k$ in $G_d$ is \emph{critical} if there does not exist a pair of separating walks of length~$k'$ in $G_d$ for any $k' < k$.

Consider the graph~$G_5$ in Figure~\ref{fig:graphg}. One example of a PSW in $G_5$ is the pair~$(\w_1,\w_2)$, where $\w_i = (v_0^i,v_1^i,\dotsc,v_7^i)$ for all $i = 1,2$, and the sequence $\pi_5(v^i_j,v^i_{j-1})$, $j = 1,2,\dotsc,7$, of generalised port numbers is 2, 2, 3, 3, 4, 4, 5. Observe that now node~$v^1_7$ has a neighbour~$v^1_{8}$ with $\pi_5(v^1_{8},v^1_7) = 5$, but node~$v^2_7$ does not have such a neighbour. The fact that the sequence grows slowly towards the parameter~$d$ is actually a general property of PSWs; this is one of the crucial ideas behind our proof.

The outline of the proof is as follows. First, we will prove auxiliary results concerning the graphs~$G_d$ and PSWs. These will enable us to obtain a lower bound for the length of PSWs. Then, we will show that this lower bound entails bisimilarity of the nodes $((1,0))$ and $((2,1))$ up to the respective distance. Since the overall proof is going to be a little hairy, we provide a chart of dependencies between the various lemmas in Figure~\ref{fig:depchart}. The first four lemmas follow quite easily from the definition of the graphs.

\begin{figure}[p]
  \captionsetup{skip=2cm}
  \centering
  \includegraphics[page=4]{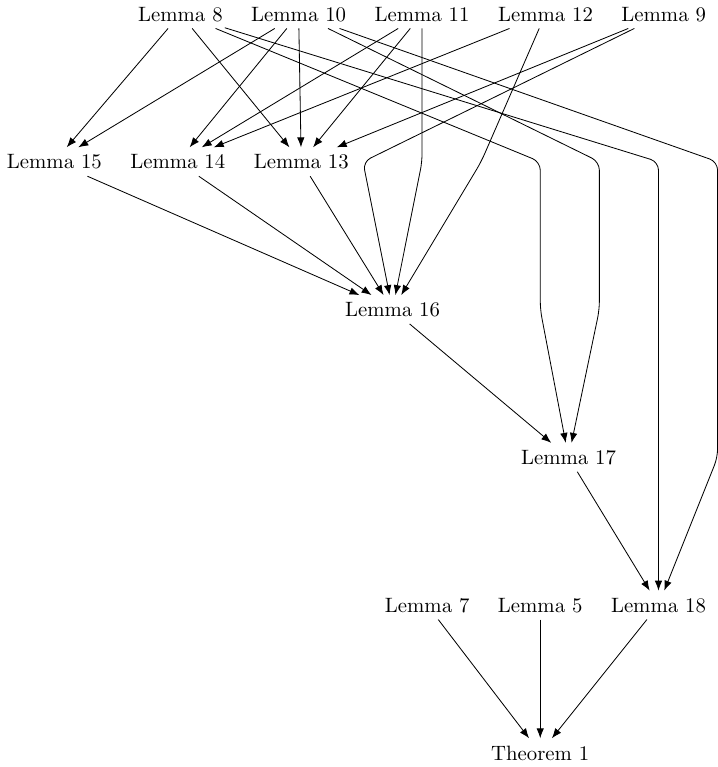}
  \caption{Dependencies between the lemmas that are needed in order to prove Theorem~\ref{thm:simulation}.}
  \label{fig:depchart}
\end{figure}

\begin{restatable}{lemma}{lemdegree}\label{lem:degree}
  For each $d$, we have $\deg(v) \in \set{1,d}$ for all $v \in V_d$, and thus $G_d \in \F(d)$. Additionally, $G_d$ is a subgraph of $G_{d+1}$.
\end{restatable}

\begin{proof}
  Consider a node $v = (a_1,a_2,\dotsc,a_i) \in V_d$. It follows from the definition that if $\set{v,u} \in E_d$, then $u$ is either a parent or a child of $v$. If $i = 0$, $v$ has no parents, and all its children are given by rule~(G2). Hence $\deg(v) = d$. If $0 < i < 2d$, $v$ has one parent, and all its children are given by rule~(G3) or (G4). There are $d-1$ children, and hence $\deg(v) = d$. If $i \ge 2d$, $v$ has no children, and hence $\deg(v) = 1$.

  It follows from the rules (G1)--(G4) that if $v \in V_d$, then also $v \in V_{d+1}$. Additionally, if $\set{v,u} \in E_d$, then clearly $\set{v,u} \in E_{d+1}$. This shows that $G_d$ is a subgraph of $G_{d+1}$.
\end{proof}

\vspace{-1.5ex}

\begin{restatable}{lemma}{lempnuniq}\label{lem:pnuniq}
  Let $v \in V_d$ and $a \in \set{0,1,\dotsc,d}$. Then there is at most one node $u \in V_d$ such that $\set{v,u} \in E_d$ and $\pi_d(u,v) = a$.
\end{restatable}

\begin{proof}
  The claim follows immediately from rule~(G2) and the way the numbers~$c_2^j$ are defined in rules (G3) and (G4).
\end{proof}

A consequence of Lemma~\ref{lem:pnuniq} is that in a walk, the successor of each node is uniquely determined by the port number from the successor to the node.

\begin{restatable}{lemma}{lempnexist}\label{lem:pnexist}
  Let $v = (a_1,a_2,\dotsc,a_i) \in V_d$, where $i < 2d$. If $v$ is odd, then for all $a \in \set{1,2,\dotsc,d}$ there exists $u \in V_d$ such that $\set{v,u} \in E_d$ and $\pi_d(u,v) = a$. If $v$ is even, then either for all $a \in \set{0,1,\dotsc,d-1}$ or for all $a \in \set{0,1,\dotsc,d-2,d}$ there exists $u \in V_d$ such that $\set{v,u} \in E_d$ and $\pi_d(u,v) = a$. In the case of even $v$ and $a = d$, node~$u$ is the parent of node~$v$.
\end{restatable}

\begin{proof}
  Observe that in rules~(G3) and (G4), we always have $b_1 \in \set{1,2,\dotsc,d}$. If $v$ is odd, the claim follows from the way the numbers~$c_2^j$ are defined in rule~(G3). If $v$ is even, consider the application of rule~(G4) to $v$. If $b_1 < d$, then $c_2^j$ will range over all the elements in $\set{0,1,\dotsc,d-1} \setminus \set{b_1}$, and thus for all $a = \set{0,1,\dotsc,d-1}$ there is a neighbour~$u$ such that $\pi_d(u,v) = a$. If $b_1 = d$, then $c_2^j$ will range over all the elements in $\set{0,1,\dotsc,d-2}$, and thus for all $a = \set{0,1,\dotsc,d-2,d}$ there is a neighbour~$u$ such that $\pi_d(u,v) = a$. We have always $c_2^j \ne d$, and hence the case $a = d$ is only possible if $b_1 = d$. It follows that if $\pi_d(u,v) = d$, $u$ is the parent of $v$.
\end{proof}

Lemma~\ref{lem:pnexist} implies that in a PSW, the last nodes of each walk must be even. Furthermore, one of the last nodes~$v$ must have a parent~$u$ with $\pi_d(u,v) = d$. It follows that we must have $v \in V_d \setminus V_{d-1}$.

\begin{restatable}{lemma}{lemnewedge}\label{lem:newedge}
  Let $\set{v,u} \in E_{d+1} \setminus E_d$ be such that $v \in V_d$. Then $u$ is a child of $v$. If $v$ is odd, then $\pi_{d+1}(v,u) = \pi_{d+1}(u,v) = d+1$. If $v$ is even, then $\pi_{d+1}(v,u) = d+1$ and $\pi_{d+1}(u,v) \in \set{d-1,d}$.
\end{restatable}

\begin{proof}
  Since $\set{v,u} \in E_{d+1}$, $u$ is either the parent or a child of $v$. If it was the parent, we would have $u \in V_d$ and thus $\set{v,u} \in E_d$, a contradiction. Hence $u \in V_{d+1} \setminus V_d$ is a child of $v$. If $v = (a_1,a_2,\dotsc,a_i)$ is odd, $u$ is given by rule~(G3) in the definition of $G_{d+1}$. Since $(a_1,a_2,\dotsc,a_{i+1}^j) \in V_d$ for all $j = 1,2,\dotsc,d-1$, we have $u = (a_1,a_2,\dotsc,a_{i+1}^d)$. As $v \in V_d$, we have $b_1,b_2^+ \le d$, and thus $c_1^d = c_2^d = d+1$. This implies $\pi_{d+1}(v,u) = \pi_{d+1}(u,v) = d+1$. If $v = (a_1,a_2,\dotsc,a_i)$ is even, $u$ is given by rule~(G4) in the definition of $G_{d+1}$. Again, we have $u = (a_1,a_2,\dotsc,a_{i+1}^d)$, $b_1,b_2 \le d$ and thus $c_1^d = d+1$. If $b_1 = d$, then $c_2^d = d-1$, otherwise $c_2^d = d$. This implies $\pi_{d+1}(v,u) = d+1$ and $\pi_{d+1}(u,v) \in \set{d-1,d}$.
\end{proof}

With the above observations out of the way, we now go forward with more powerful results.

\begin{restatable}{lemma}{lemidedges}\label{lem:idedges}
  Let $(\w_1,\w_2)$, where $\w_i = (v_0^i,v_1^i,\dotsc,v_k^i)$ for all $i = 1,2$, be a PSW in $G_d$. If for some $\ell \in \set{0,1,\dotsc,k-1}$ the node~$v_{\ell+1}^i$ is a child of node~$v_\ell^i$ for all $i = 1,2$, and we have $\pi_d(v_\ell^1,v_{\ell+1}^1) = \pi_d(v_\ell^2,v_{\ell+1}^2)$, then $(\w_1,\w_2)$ is not a critical PSW in $G_d$.
\end{restatable}

\begin{proof}
  Suppose that for all $m = \ell+2,\ell+3,\dotsc,k$ we have $v_m^1 \ne v_\ell^1$ or $v_m^2 \ne v_\ell^2$. By assumption, $v_{\ell+1}^1$ and $v_{\ell+1}^2$ are of the same type. Consider the definition of $G_d$. Now it is easy to show by induction on $m$ that nodes $v_m^1$ and $v_m^2$ are of the same type for all $m = \ell+1,\ell+2,\dotsc,k$. Since $k \le 2d-3$, both $v_k^1$ and $v_k^2$ have child nodes. It follows that if $v_{k+1}^1$ is a neighbour of $v_k^1$, there is a neighbour~$v_{k+1}^2$ of $v_k^2$ such that $\pi_d(v_{k+1}^1,v_k^1) = \pi_d(v_{k+1}^2,v_k^2)$. Thus $(\w_1,\w_2)$ is not a PSW in $G_d$, a contradiction.

  Now $v_m^1 = v_\ell^1$ and $v_m^2 = v_\ell^2$ for some $m \in \set{\ell+2,\ell+3,\dotsc,k}$. Let
  \[
    \w'_i = (v_0^i,v_1^i,\dotsc,v_\ell^i,v_{m+1}^i,v_{m+2}^i,\dotsc,v_k^i)
  \]
  for all $i = 1,2$. Then $(\w'_1,\w'_2)$ is a PSW of length $k-m+\ell \le k-(\ell+2)+\ell = k-2 < k$ in $G_d$ and hence $(\w_1,\w_2)$ is not critical.
\end{proof}

\vspace{-1.5ex}

\begin{restatable}{lemma}{lemextendpsw}\label{lem:extendpsw}
  Let $(\w_1,\w_2)$ be a PSW of length~$k$ in $G_d$. Then there is a PSW of length~$k+2$ in $G_{d+1}$.
\end{restatable}

\begin{proof}
  Let $\w_i = (v_0^i,v_1^i,\dotsc,v_k^i)$ for all $i = 1,2$. By definition, there is a neighbour~$u \in V_d$ of $v_k^1$ such that for each neighbour~$w \in V_d$ of $v_k^2$ we have $\pi_d(u,v_k^1) \ne \pi_d(w,v_k^2)$. Lemma~\ref{lem:pnexist} implies that $v_k^1$ and $v_k^2$ are even, $\pi_d(u,v_k^1) \in \set{d-1,d}$, and there is a neighbour~$w \in V_d$ of $v_k^2$ for which $\pi_d(w,v_k^2) \in \set{d-1,d} \setminus \set{\pi_d(u,v_k^1)}$. That is, we have $\pi_d(u,v_k^1) = d$ or $\pi_d(w,v_k^2) = d$. Without loss of generality, we can assume $\pi_d(u,v_k^1) = d$ and thus $\pi_d(w,v_k^2) = d-1$.

  Lemma~\ref{lem:degree} implies that $\deg_{G_d}(u) = \deg_{G_d}(v_k^2) = d$ and $\deg_{G_{d+1}}(u) = \deg_{G_{d+1}}(v_k^2) = d+1$. Hence there are nodes $x,y \in V_{d+1} \setminus V_d$ such that $\set{u,x} \in E_{d+1} \setminus E_d$ and $\set{v_k^2,y} \in E_{d+1} \setminus E_d$. Note that $u,v_k^2 \in V_d$, $u$ is odd and $v_k^2$ is even. It follows from Lemma~\ref{lem:newedge} that $\pi_{d+1}(u,x) = \pi_{d+1}(x,u) = d+1$, $\pi_{d+1}(v_k^2,y) = d+1$ and $\pi_{d+1}(y,v_k^2) \in \set{d-1,d}$. Since $\pi_{d+1}(w,v_k^2) = \pi_d(w,v_k^2) = d-1$ and $w \ne y$, Lemma~\ref{lem:pnuniq} implies that $\pi_{d+1}(y,v_k^2) = d$.

  Now we can extend the walks $\w_1$ and $\w_2$. Set $\w'_1 = (v_0^1,v_1^1,\dotsc,v_k^1,u,x)$ and $\w'_2 = (v_0^2,v_1^2,\dotsc,\allowbreak v_k^2,y,v_k^2)$. We have $\pi_{d+1}(u,v_k^1) = d = \pi_{d+1}(y,v_k^2)$ and $\pi_{d+1}(x,u) = d+1 = \pi_{d+1}(v_k^2,y)$, as required. Furthermore, node~$x$ has neighbour~$u$ for which $\pi_{d+1}(u,x) = d+1$. Suppose that there is a neighbour~$u'$ of $v_k^2$ for which $\pi_{d+1}(u',v_k^2) = d+1$. Now Lemma~\ref{lem:pnexist} implies that $u'$ is the parent of $v_k^2$. But since $v_k^2 \in V_d$, we have also $u' \in V_d$, and hence $\pi_{d+1}(u',v_k^2) \le d$, a contradiction. Similarly, node~$v_k^2$ has neighbour~$y$ for which $\pi_{d+1}(y,v_k^2) = d$, but $\pi_{d+1}(u,x) = d+1$ together with Lemma~\ref{lem:pnexist} implies that there is no neighbour~$y'$ of $x$ for which $\pi_{d+1}(y',x) = d$. This shows that $(\w'_1,\w'_2)$ is a PSW of length~$k+2$ in $G_{d+1}$.
\end{proof}

\vspace{-1.5ex}

\begin{restatable}{lemma}{lempswnewnode}\label{lem:pswnewnode}
  Let $(\w_1,\w_2)$, where $\w_i = (v_0^i,v_1^i,\dotsc,v_k^i)$ for all $i = 1,2$, be a critical PSW in $G_d$. Then we have $v_{k-1}^i \in V_d \setminus V_{d-1}$ for some $i \in \set{1,2}$.
\end{restatable}

\begin{proof}
  Lemma~\ref{lem:pnexist} implies that $v_k^1$ and $v_k^2$ are even, and for some $i \in \set{1,2}$ node~$v_k^i$ has a parent~$u$ such that $\pi_d(u,v_k^i) = d$. If $v_k^i \in V_{d-1}$, then also $u \in V_{d-1}$ and hence $\pi_d(u,v_k^i) \le d-1$, a contradiction. Therefore $v_k^i \in V_d \setminus V_{d-1}$.

  Suppose that $v_{k-1}^j \in V_{d-1}$ for all $j = 1,2$. Since $v_k^i \in V_d \setminus V_{d-1}$, we have $\set{v_{k-1}^i,v_k^i} \in E_d \setminus E_{d-1}$. Lemma~\ref{lem:newedge} implies that $v_k^i$ is a child of $v_{k-1}^i$ and $\pi_d(v_{k-1}^i,v_k^i) = \pi_d(v_k^i,v_{k-1}^i) = d$. Let $j \in \set{1,2} \setminus \set{i}$. As $\pi_d(v_k^j,v_{k-1}^j) = \pi_d(v_k^i,v_{k-1}^i) = d$, we have $\set{v_{k-1}^j,v_k^j} \in E_d \setminus E_{d-1}$ and thus $v_k^j$ is a child of $v_{k-1}^j$ and $\pi_d(v_{k-1}^j,v_k^j) = \pi_d(v_k^j,v_{k-1}^j) = d$. Now it follows from Lemma~\ref{lem:idedges} that $(\w_1,\w_2)$ is not a critical PSW in $G_d$, a contradiction.
\end{proof}

\vspace{-1.5ex}

\begin{restatable}{lemma}{lempswsym}\label{lem:pswsym}
  Let $(\w_1,\w_2)$, where $\w_i = (v_0^i,v_1^i,\dotsc,v_k^i)$ for all $i = 1,2$, be a PCW in $G_d$. If $(\w_1,\w_2)$ is not a PSW in $G_d$, then for each neighbour~$v_{k+1}^1 \in V_d$ of $v_k^1$ there is a neighbour~$v_{k+1}^2 \in V_d$ of $v_k^2$ such that $\pi_d(v_{k+1}^1,v_k^1) = \pi_d(v_{k+1}^2,v_k^2)$, and vice versa.
\end{restatable}

\begin{proof}
  Since $(\w_1,\w_2)$ is not a PSW, condition (W3) does not hold. This is equivalent to the first claim. For the second claim, assume that $v_{k+1}^2$ is a neighbour of $v_k^2$. Suppose that there is no neighbour~$v_{k+1}^1$ of $v_k^1$ such that $\pi_d(v_{k+1}^1,v_k^1) = \pi_d(v_{k+1}^2,v_k^2)$. Now it follows from Lemma~\ref{lem:degree} and Lemma~\ref{lem:pnexist} that $v_k^1$ and $v_k^2$ are even and $\pi_d(v_{k+1}^2,v_k^2) \in \set{d-1,d}$. We also obtain from Lemma~\ref{lem:pnexist} that there is a neighbour~$u$ of $v_k^1$ for which $\pi_d(u,v_k^1) \in \set{d-1,d} \setminus \set{\pi_d(v_{k+1}^2,v_k^2)}$. Now $u$ is a neighbour of $v_k^1$ such that there is no neighbour~$w$ of $v_k^2$ for which $\pi_d(u,v_k^1) = \pi_d(w,v_k^2)$, a contradiction.
\end{proof}

Now we are ready to prove the following lemma, which is the main ingredient of the proof of Theorem~\ref{thm:simulation}. The underlying idea is that the generalised port numbers along the walks have to grow slowly. Put otherwise, each prefix of a critical PSW must be contained in a subgraph~$G_d$ for a sufficiently small value of $d$.

\begin{lemma}\label{lem:prefixpsw}
  Let $(\w_1,\w_2)$, where $\w_i = (v_0^i,v_1^i,\dotsc,v_k^i)$ for all $i = 1,2$, be a critical PSW in $G_d$. Then $(\w'_1,\w'_2)$, where $\w'_i = (v_0^i,v_1^i,\dotsc,v_{k-2}^i)$ for all $i = 1,2$, is a PSW in $G_{d-1}$.
\end{lemma}

\begin{proof}
  First, suppose that $\set{v_\ell^i,v_{\ell+1}^i} \in E_{d-1}$ for all $i = 1,2$ and $\ell = 0,1,\dotsc,k-3$ but that $(\w'_1,\w'_2)$ is not a PSW in $G_{d-1}$. Assume that $\set{v_{k-2}^i,v_{k-1}^i} \in E_{d-1}$ for some $i \in \set{1,2}$ and let $j \in \set{1,2} \setminus \set{i}$. It follows from Lemma~\ref{lem:pswsym} that there is a neighbour~$u \in V_{d-1}$ of $v_{k-2}^j$ such that $\pi_{d-1}(u,v_{k-2}^j) = \pi_{d-1}(v_{k-1}^i,v_{k-2}^i)$. Now Lemma~\ref{lem:pnuniq} implies that $u = v_{k-1}^j$ and hence we have $v_{k-1}^i,v_{k-1}^j \in V_{d-1}$. Then we can use Lemma~\ref{lem:pswnewnode} to obtain that $(\w_1,\w_2)$ is not a critical PSW in $G_d$, a contradiction.

  Let us then assume that $\set{v_{k-2}^i,v_{k-1}^i} \in E_d \setminus E_{d-1}$ for all $i = 1,2$. As $v_{k-2}^i \in V_{d-1}$ for all $i = 1,2$, Lemma~\ref{lem:newedge} implies that $v_{k-1}^i$ is a child of $v_{k-2}^i$ and $\pi_d(v_{k-2}^1,v_{k-1}^1) = d = \pi_d(v_{k-2}^2,v_{k-1}^2)$ for all $i = 1,2$. But now we can apply Lemma~\ref{lem:idedges} to see that $(\w_1,\w_2)$ is not a critical PSW in $G_d$, a contradiction. We have now shown that if $\set{v_\ell^i,v_{\ell+1}^i} \in E_{d-1}$ for all $i = 1,2$ and $\ell = 0,1,\dotsc,k-3$, then $(\w'_1,\w'_2)$ is a PSW in $G_{d-1}$.

  Then, suppose that $\set{v_\ell^i,v_{\ell+1}^i} \in E_d \setminus E_{d-1}$ for some $i \in \set{1,2}$ and $\ell \in \set{0,1,\dotsc,k-3}$. Let $m$ be the smallest value of $\ell$ for which this holds. Let $j \in \set{1,2} \setminus \set{i}$. If $m$ is even, then the node~$v_m^i \in V_{d-1}$ is odd, and by Lemma~\ref{lem:newedge} we have that $\pi_d(v_m^i,v_{m+1}^i) = \pi_d(v_{m+1}^i,v_m^i) = d$ and that $v_{m+1}^i$ is a child of $v_m^i$. Since $\pi_d(v_{m+1}^j,v_m^j) = \pi_d(v_{m+1}^i,v_m^i) = d$, we obtain $\set{v_m^j,v_{m+1}^j} \in E_d \setminus E_{d-1}$. As $v_m^j \in V_{d-1}$ is odd, Lemma~\ref{lem:newedge} yields that $\pi_d(v_m^j,v_{m+1}^j) = \pi_d(v_{m+1}^j,v_m^j) = d$ and that $v_{m+1}^j$ is a child of $v_m^j$. Lemma~\ref{lem:idedges} then implies that $(\w_1,\w_2)$ is not a critical PSW in $G_d$, a contradiction.

  To complete the proof, assume that $m$ is odd. Recall that $\set{v_m^i,v_{m+1}^i} \in E_d \setminus E_{d-1}$. If also $\set{v_m^j,v_{m+1}^j} \in E_d \setminus E_{d-1}$, we can again use Lemma~\ref{lem:newedge} to get that $v_{m+1}^i$ and $v_{m+1}^j$ are children of $v_m^i$ and $v_m^j$, respectively, and that $\pi_d(v_m^i,v_{m+1}^i) = d = \pi_d(v_m^j,v_{m+1}^j)$. Now Lemma~\ref{lem:idedges} yields a contradiction. If $\set{v_m^j,v_{m+1}^j} \in E_{d-1}$, let $\w''_\ell = (v_0^\ell,v_1^\ell,\dotsc,v_m^\ell)$ for all $\ell = 1,2$. The pair $(\w''_1,\w''_2)$ is a PSW in $G_{d-1}$, because otherwise by using a similar argument as above we would obtain that $\set{v_m^i,v_{m+1}^i} \in E_{d-1}$, a contradiction. But now we can use Lemma~\ref{lem:extendpsw} to get a PSW of length $m+2 \le (k-3)+2 = k-1$ in $G_d$, which contradicts the criticality of $(\w_1,\w_2)$.
\end{proof}

Having proved Lemma~\ref{lem:prefixpsw}, the following result now follows by induction.

\begin{restatable}{lemma}{lempswlen}\label{lem:pswlen}
  Let $(\w_1,\w_2)$ be a PSW of length~$k$ in $G_d$. Then $k \ge 2d-3$.
\end{restatable}

\begin{proof}
  We use induction on $d$. Let $\w_i = (v_0^i,v_1^i,\dotsc,v_k^i)$ for all $i = 1,2$. It follows from Lemma~\ref{lem:degree} and Lemma~\ref{lem:pnexist} that $v_k^i$ is even for all $i = 1,2$, and thus $k$ is odd. Hence we have $k \ge 1$. If $d = 2$, we have shown that $k \ge 2d-3$.

  For the inductive step, suppose that the claim holds for $d = q$ and that $(\w_1,\w_2)$ is a PSW of length $k$ in $G_{q+1}$. Now there is a critical PSW $(\ww_1,\ww_2)$ of length $\ell \le k$ in $G_{q+1}$, where $\ww_i = (u_0^i,u_1^i,\dotsc,u_\ell^i)$ for all $i = 1,2$. Lemma~\ref{lem:prefixpsw} implies that $(\ww'_1,\ww'_2)$, where $\ww'_i = (u_0^i,u_1^i,\dotsc,u_{\ell-2}^i)$ for all $i = 1,2$, is a PSW of length $\ell-2$ in $G_q$. By the inductive hypothesis we obtain $\ell-2 \ge 2q-3$. It follows that $k \ge \ell \ge 2q-1 = 2(q+1)-3$. Hence we have shown that the claim holds for $d = q+1$.
\end{proof}

Now we just need to show that the lower bound for the length of PSWs implies bisimilarity up to the respective distance, and we are mostly done.

\begin{restatable}{lemma}{lempswtobisim}\label{lem:pswtobisim}
  We have $((1,0)) \nbsmsv{2d-3} ((2,1))$, that is, the nodes $((1,0))$ and $((2,1))$ of $G_d$ are $(2d-3)$-$\aSV$-bisimilar.
\end{restatable}

\begin{proof}
  If we have $((1,0)) \nbsmsv{k} ((2,1))$ for arbitrarily large $k$, the claim is clearly true. Otherwise, let $k$ be the largest integer for which we have $((1,0)) \nbsmsv{k} ((2,1))$. We will show that $k \ge 2d-3$.

  Let $v_0^1 = ((1,0))$ and $v_0^2 = ((2,1))$. Suppose then that $\ell \in \set{0,1,\dotsc,k-1}$ and that $v_\ell^1$ and $v_\ell^2$ have been defined. Furthermore, suppose that $k-\ell$ is the largest integer~$m$ for which $v_\ell^1 \nbsmsv{m} v_\ell^2$ holds. If for each neighbour~$u$ of $v_\ell^1$ there was a neighbour~$w$ of $v_\ell^2$, and vice versa, such that $u \nbsmsv{k-\ell} w$ and $\pi_d(u,v_\ell^1) = \pi_d(w,v_\ell^2)$, then by Definition~\ref{def:nbisimsv} we would have $v_\ell^1 \nbsmsv{k-\ell+1} v_\ell^2$, a contradiction. Thus for some $i \in \set{1,2}$ and $j \in \set{1,2} \setminus \set{i}$ there is a neighbour~$u$ of $v_\ell^i$ such that there is no neighbour~$w$ of $v_\ell^j$ for which the given condition holds. However, since $v_\ell^i \nbsmsv{k-\ell} v_\ell^j$, we can choose neighbour~$w$ so that $u \nbsmsv{k-\ell-1} w$ and $\pi_d(u,v_\ell^i) = \pi_d(w,v_\ell^j)$. Now we can define $v_{\ell+1}^i = u$ and $v_{\ell+1}^j = w$. We have shown that $k-\ell-1 = k-(\ell+1)$ is the largest integer~$m$ for which $v_{\ell+1}^i \nbsmsv{m} v_{\ell+1}^j$ holds.

  The above recursive definition yields a pair~$(\w_1,\w_2)$ of walks, where $\w_i = (v_0^i,v_1^i,\dotsc,v_k^i)$ for all $i = 1,2$. Clearly conditions (W1) and (W2) hold. Additionally, we know that $k-k = 0$ is the largest integer~$m$ for which we have $v_k^1 \nbsmsv{m} v_k^2$. However, if $k \le 2d-3$, then for each neighbour~$u$ of $v_k^1$ and $w$ of $v_k^2$ we have $\deg(u) = \deg(w)$ and hence $u \nbsmsv{0} w$. It follows that for some $i \in \set{1,2}$ and $j \in \set{1,2} \setminus \set{i}$ there is a neighbour~$u$ of $v_k^i$ such that there is no neighbour~$w$ of $v_k^j$ for which $\pi_d(u,v_k^i) = \pi_d(w,v_k^j)$. If $i = 1$ and $j = 2$, this is equivalent to condition~(W3). Otherwise, we use Lemma~\ref{lem:degree} and Lemma~\ref{lem:pnexist} to swap the roles of $i$ and $j$ in a similar manner as in the proof of Lemma~\ref{lem:pswsym}.

  In conclusion, we have shown that if $k \le 2d-3$, then $(\w_1,\w_2)$ is a PSW of length~$k$ in $G_d$. Now Lemma~\ref{lem:pswlen} implies that $k = 2d-3$. If $k > 2d-3$, the claim is trivially true.
\end{proof}

\begin{remark}
  Lemma~\ref{lem:pswtobisim} can also be viewed from a game-theoretic perspective. When considering a game played by \emph{Spoiler} and $\emph{Duplicator}$ starting from the nodes $((1,0))$ and $((2,1))$, the pair of sequences consisting of the nodes chosen by the players is a PSW. Then, the lower bound on the length of PSWs implies that Duplicator has a \emph{winning strategy} in the $(2d-3)$-round bisimulation game. For more details on bisimulation games, see \textcite{blackburn07handbook}.
\end{remark}

To prove Theorem~\ref{thm:simulation}, we want the root node~$\emptyset$ to receive the same messages from its neighbours $((1,0))$ and $((2,1))$. Lemma~\ref{lem:pswtobisim} shows that they are $(2d-3)$-$\aSV$-bisimilar, but this is not enough: they also need to have identical outgoing port numbers towards node~$\emptyset$. We will now define a port numbering of $G_d$ based on the generalised port numbering~$p_d$. Let $f \colon \set{0,1,\dotsc,d} \to [d]$ be a function such that $f(0) = 1$ and $f(i) = i$ for $i \ge 1$. If $p_d(v,i) = (u,j)$ for some nodes $v,u$ and port numbers $i,j$, we define $p'_d(v,(f(i)) = (u,f(j))$. Due to the fact that in rule~(G3) of the definition of $G_d$ we used $b_2^+$ instead of $b_2$, no node has both $0$ and $1$ as port numbers in $p_d$. It follows that $p'_d$ is a bijection from the set of input ports to the set of output ports, and the set of outgoing as well as incoming port numbers for each node~$v$ is $\set{1,2,\dotsc,\deg(v)}$. Observe that $p'_d(((1,0)),1) = (\emptyset,1)$ and $p'_d(((2,1)),1) = (\emptyset,2)$. Now we can apply Lemma~\ref{lem:pnchangebsm} to see that the $(2d-3)$-$\aSV$-bisimilarity still holds, that is, we have $(G_d,((1,0)),p'_d) \nbsmsv{2d-3} (G_d,((2,1)),p'_d)$.

Let $\aA \in \aSV$ be an arbitrary algorithm and $\Delta \ge 2$. Let $G = G_\Delta$, $p = p'_\Delta$, $v = \emptyset$, $u = ((1,0))$ and $w = ((2,1))$. Consider the execution of $\aA$ in $(G,p)$. Lemma~\ref{lem:nbsmsvstate} implies that the state of $\aA$ in the nodes $u$ and $w$ is identical in each round $r = 0,1,\dotsc,2\Delta-3$. Furthermore, we have $\pi(u,v) = 1 = \pi(w,v)$. It follows that $u$ and $w$ send the same message to node~$v$ in each round $r+1 = 1,2,\dotsc,2\Delta-2$. This concludes the proof of Theorem~\ref{thm:simulation}.

\begin{remark}
  We could as well show that the nodes $((1,0))$ and $((2,1))$ are $(2d-3)$-bisimilar with respect to the class $\aMV$ of algorithms, with only minor changes to the proof of Lemma~\ref{lem:pswtobisim}. However, this would not make any difference in the end, since we need to consider an algorithm in $\aSV$ for the root node to lose the multiplicities of messages it receives from its neighbours.
\end{remark}

\section{Separation by a Graph Problem}\label{sec:graphlb}

Theorem~\ref{thm:simulation} shows that the simulation algorithm is optimal in a certain sense. However, since we are interested in graph problems, we want to separate the classes $\aSV$ and $\aMV$ by one. The following theorem states that we can do this, and the lower bound in still linear in $\Delta$.

\begin{thm2}
  There is a graph problem~$\Pi$ that can be solved in one round by an algorithm in $\sMV$ but that requires at least time $T$, where $T(n,\Delta) \ge \Delta$ for all $\Delta \ge 2$, when solved by an algorithm in $\sSV$.
\end{thm2}

Let us first define formally the graph problem~$\Pi$. We will be working with graphs where each node is given as a local input one of three colours: black ($\cB$), white ($\cW$) or grey ($\cG$). For each graph~$(G,f)$ with local input from the set $\set{\cB,\cW,\cG}$, the set~$\Pi(G,f)$ of solutions consists of mappings $S \colon V \to \set{\cB,\cW,\cG}$ such that for each $v \in V$, $S(v)$ is one of the local inputs having the highest multiplicity among the neighbours of $v$. For example, if node~$v$ has four neighbours of colour~$\cB$, four neighbours of colour~$\cW$ and two neighbours of colour~$\cG$, then for each solution~$S$ we have $S(v) = \cB$ or $S(v) = \cW$.

There is an algorithm in $\sMV$---and, in fact, in $\sMB$---that solves problem~$\Pi$ in only one communication round: Each node broadcasts its own colour to all its neighbours. Then, each node counts the multiplicity of each message it received and outputs the one with the highest multiplicity. Showing that this cannot be solved by any algorithm in $\sSV$ in less than $\Delta$ communication rounds will require somewhat more work. Luckily, we can handle the most tricky part of the proof by making use of the proof of Theorem~\ref{thm:simulation} in a black-box manner.

We start by defining for each $d = 2,3,\dotsc$ two graphs, $H_{\cB,d} = (V_{\cB,d},E_{\cB,d})$ and $H_{\cW,d} = (V_{\cW,d},E_{\cW,d})$. The constructions can be seen as extensions of the graph~$G_d$ defined earlier, but now each node is coloured with one of the three colours: black ($\cB$), white ($\cW$) or grey ($\cG$). Colours $\cB$ and $\cW$ can be thought of as complements of each other; we write $\overline{\cB} = \cW$ and $\overline{\cW} = \cB$. Again, we define $V_{\cB,d}$ recursively:

\begin{enumerate}[label=(H\arabic*)]
  \item $\emptyset \in V_{\cB,d}.$
  \item $((1,0,\cB)), ((2,1,\cB)), ((3,2,\cB)), ((4,3,\cB)),\dotsc,((d,d-1,\cB)) \in V_{\cB,d}.$
  \item $((2,1,\cW)), ((3,2,\cW)), ((4,3,\cW)),\dotsc,((d,d-1,\cW)) \in V_{\cB,d}.$
  \item If $(a_1,a_2,\dotsc,a_i) \in V_{\cB,d}$, where $i$ is odd and $i < 2d$, then $(a_1,a_2,\dotsc,a_{i+1}^j) \in V_{\cB,d}$ for all $j = 1,2,\dotsc,d-1$, where $a_{i+1}^j = (c_1^j, c_2^j, \cG)$ is defined as follows. Let $(b_1,b_2,C) = a_i$, where $C \in \set{\cB,\cW}$, and $b_2^+ = 1$ if $b_2 = 0$, $b_2^+ = b_2$ otherwise. Define
    \begin{align*}
      c_1^j &= \min(\set{1,2,\dotsc,d} \setminus \set{b_2^+,c_1^1,c_1^2,\dotsc,c_1^{j-1}}),\\
      c_2^j &= \min(\set{1,2,\dotsc,d} \setminus \set{b_1,c_2^1,c_2^2,\dotsc,c_2^{j-1}}).
    \end{align*}
  \item If $(a_1,a_2,\dotsc,a_i) \in V_{\cB,d}$, where $i$ is even and $i < 2d$, then $(a_1,a_2,\dotsc,a_{i+1}^j) \in V_{\cB,d}$ for all $j = 1,2,\dotsc,d-1$, where $a_{i+1}^j = (c_1^j, c_2^j,C)$ is defined as follows. Let $(d_1,d_2,C) = a_{i-1}$, where $C \in \set{\cB,\cW}$, and $(b_1,b_2,\cG) = a_i$. Define
    \begin{align*}
      c_1^j &= \min(\set{1,2,\dotsc,d} \setminus \set{b_2,c_1^1,c_1^2,\dotsc,c_1^{j-1}}),\\
      c_2^j &= \min(\set{0,1,\dotsc,d-1} \setminus \set{b_1,c_2^1,c_2^2,\dotsc,c_2^{j-1}}).
    \end{align*}
  \item If $(a_1,a_2,\dotsc,a_i) \in V_{\cB,d}$, where $i$ is even and $i < 2d$, then $(a_1,a_2,\dotsc,a_{i+1}^j) \in V_{\cB,d}$ for all $j = 1,2,\dotsc,d-1$, where $a_{i+1}^j = (c_1^j, c_2^j,\overline{C})$ is defined as follows. Let $(d_1,d_2,C) = a_{i-1}$, where $C \in \set{\cB,\cW}$. Define
    \begin{align*}
      c_1^j &= \min(\set{2,3,\dotsc,d} \setminus \set{c_1^1,c_1^2,\dotsc,c_1^{j-1}}),\\
      c_2^j &= \min(\set{1,2,\dotsc,d-1} \setminus \set{c_2^1,c_2^2,\dotsc,c_2^{j-1}}).
    \end{align*}
\end{enumerate}
The set $E_{\cB,d}$ of edges consists of all pairs $\set{v,u}$, where $v = (a_1,a_2,\dotsc,a_i) \in V_{\cB,d}$ and $u = (a_1,a_2,\dotsc,a_i,a_{i+1}) \in V_{\cB,d}$ for some $i \in \set{0,1,\dotsc}$. The sets $V_{\cW,d}$ and $E_{\cW,d}$ are given by the same definition by replacing every occurrence of $\cB$ with $\cW$ and vice versa. See Figures~\ref{fig:graphhb} and~\ref{fig:graphhw} for illustrations. By rearranging the branches of the trees, we observe that actually the only difference between $H_{\cB,d}$ and $H_{\cW,d}$ is the colours in the branch that starts with the node~$((1,0,C))$.

\begin{figure}[p]
  \captionsetup{skip=2cm}
  \centering
  \includegraphics[page=2,scale=0.85]{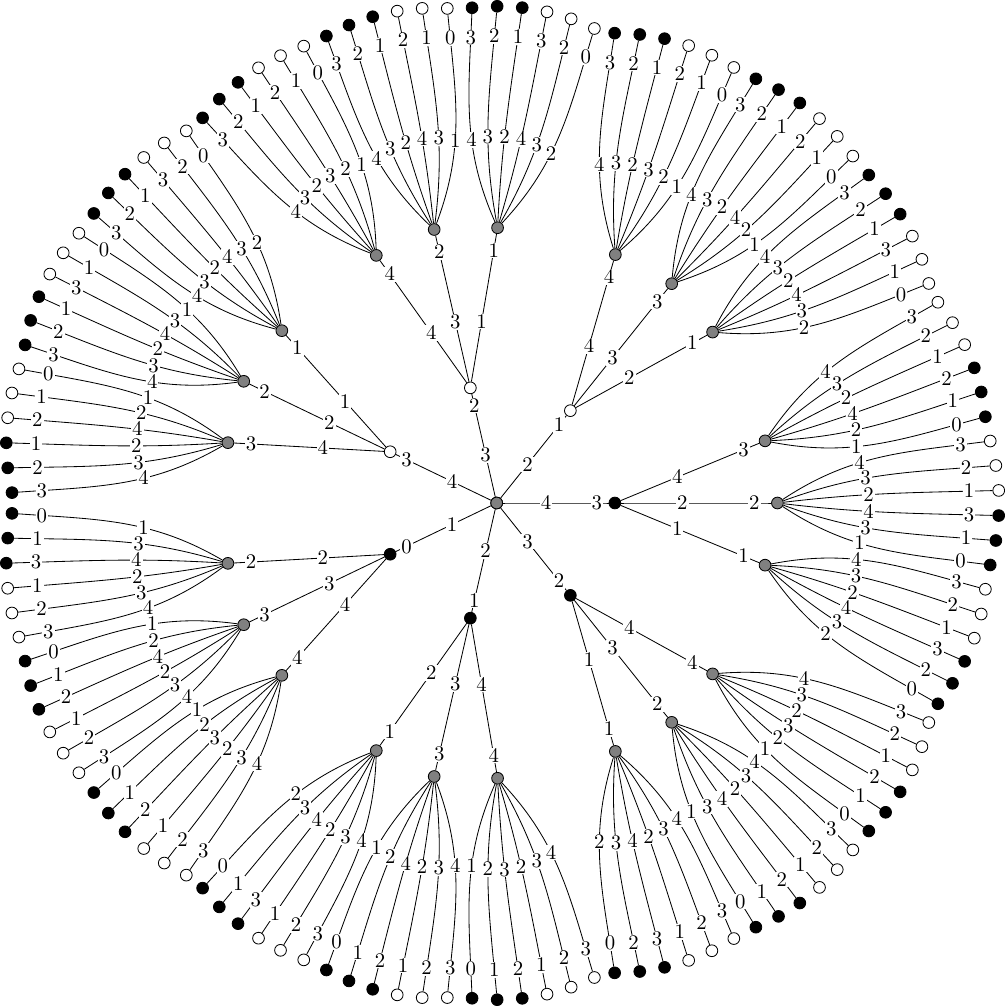}
  \caption{A part of the graph~$H_{\cB,4}$. The node in the centre is node~$\emptyset$. The numbers pictured are outgoing port numbers.}
  \label{fig:graphhb}
\end{figure}
\begin{figure}[p]
  \captionsetup{skip=2cm}
  \centering
  \includegraphics[page=3,scale=0.85]{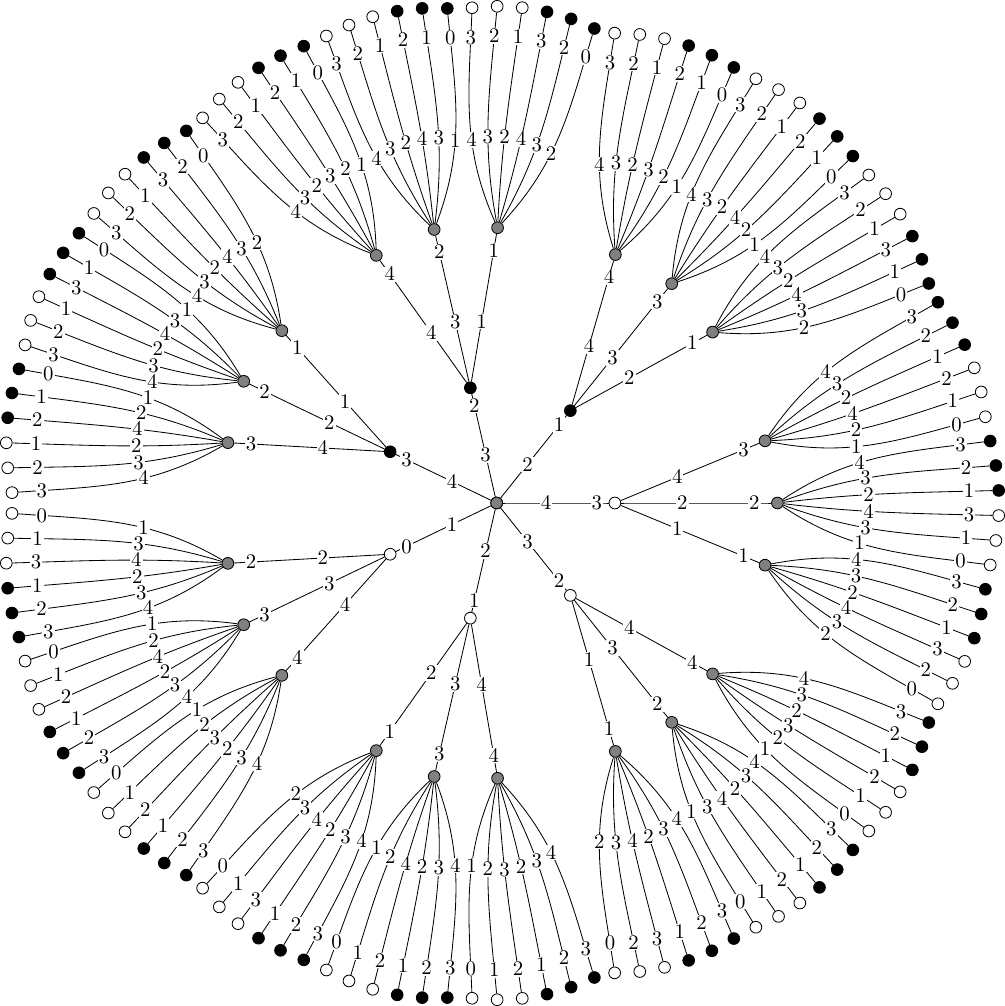}
  \caption{A part of the graph~$H_{\cW,4}$. The node in the centre is node~$\emptyset$. The numbers pictured are outgoing port numbers.}
  \label{fig:graphhw}
\end{figure}

In this proof we work with the graphs $H_{\cB,d}$ and $H_{\cW,d}$ for a fixed value of $d$. Hence, to simplify notation, we will write $H_\cB$ and $H_\cW$ from now on.

We define colourings $f_\cB \colon V_\cB \to \set{\cB,\cW,\cG}$ and $f_\cW \colon V_\cW \to \set{\cB,\cW,\cG}$ as follows. If $v = (a_1,a_2,\dotsc,a_i) \in V_C$ for some $C \in \set{\cB,\cW}$ and $i \ge 1$, and we have $a_i = (b_1,b_2,C')$, set $f_C(v) = C'$. If $v = \emptyset \in V_C$, set $f_C(v) = \cG$. Notice that for each solution $S \in \Pi(H_\cB,f_\cB)$ we have $S(\emptyset) = \cB$ and for each solution $S \in \Pi(H_\cW,f_\cW)$ we have $S(\emptyset) = \cW$.

Our port numbers are pairs $(a,C)$, where $a \in \set{0,1,\dotsc,d}$ and $C \in \set{\cB,\cW,\cG}$. Generalised port numberings $p_\cB$ and $p_\cW$ for $H_\cB$ and $H_\cW$, respectively, are defined as follows. Let $v = (a_1,a_2,\dotsc,a_i)$ and $u = (a_1,a_2,\dotsc,a_{i+1})$, where $a_{i+1} = (b_1,b_2,C)$, be nodes. Note that $f_\cB(u) = f_\cW(u) = C$. If $C \in \set{\cB,\cW}$, define
\begin{align*}
  p_\cB(v,(b_1,C)) &= p_\cW(v,(b_1,C)) = (u,(b_2,\cG)),\\
  p_\cB(u,(b_2,\cG)) &= p_\cW(u,(b_2,\cG)) = (v,(b_1,C)).
\end{align*}
If $C = \cG$, let $C' = f_\cB(v) = f_\cW(v)$ and define
\begin{align*}
  p_\cB(v,(b_1,\cG)) &= p_\cW(v,(b_1,\cG)) =  (u,(b_2,C')),\\
  p_\cB(u,(b_2,C')) &= p_\cW(u,(b_2,C')) =  (v,(b_1,\cG)).
\end{align*}

Next we will define induced subgraphs $\hat{H}_\cB$ and $\hat{H}_\cW$ of $H_\cB$ and $H_\cW$, respectively. For $C \in \set{\cB,\cW}$, the vertex set $\hat{V}_C$ of $\hat{H}_C$ consists of all vertices $(a_1,a_2,\dotsc,a_i) \in V_C$ such that $f_C((a_1,a_2,\dotsc,a_j)) \in \set{C,\cG}$ for all $j \in \set{0,1,\dotsc,i}$. That is, a node~$v$ of $H_C$ is in the subgraph~$\hat{H}_C$ if and only if each node in the unique path from the root node~$\emptyset$ to node~$v$ is either grey or of colour $C$. For each $v = (a_1,a_2,\dotsc,a_i) \in V_C$ we denote the corresponding node of $\hat{H}_C$ by $\hat{v} = (a_1,a_2,\dotsc,a_i) \in \hat{V}_C$.

For each $C \in \set{\cB,\cW}$, define a mapping $g_C \colon \hat{V}_C \to V_d$ as follows. Assume $\hat{v} = (a_1,a_2,\dotsc,a_i) \in \hat{V}_C$, where $a_j = (b_1^j,b_2^j,C_j)$ for each $j$. Now set $g_C(\hat{v}) = (a'_1,a'_2,\dotsc,a'_i)$, where $a'_j = (b_1^j,b_2^j)$ for each $j$. By observing that the subgraph~$\hat{H}_C$ is given by the rules (H1), (H2), (H4) and (H5) in the definition of $H_C$, and how they correspond to the rules (G1)--(G4) in the definition of $G_d$, one can see that $g_C$ is a bijection, and in fact an isomorphism, between $\hat{H}_C$ and $G_d$. We can use $g_C$ to move bisimilarity results from $G_d$ to $\hat{H}_C$, as the following lemma shows.

\begin{lemma}\label{lem:bsmgtosubh}
  Let $C \in \set{\cB,\cW}$, $r \in \N$ and $\hat{v},\hat{u} \in \hat{V}_C$. If $g_C(\hat{v}) \nbsmsv{r} g_C(\hat{u})$ and $f_C(\hat{v}) = f_C(\hat{u})$, then $\hat{v} \nbsmsv{r} \hat{u}$.
\end{lemma}

\begin{proof}
  The proof is by induction on $r$. Given the inductive hypothesis and conditions (B1)--(B3) of Definition~\ref{def:nbisimsv} for $g_C(\hat{v})$ and $g_C(\hat{u})$, it is quite straightforward to check that the conditions also hold for $\hat{v}$ and $\hat{u}$.
\end{proof}

Next, we will define a partial mapping $f_{v,u} \colon V_C \to V_C$ for each pair of grey nodes $\hat{v}$ and $\hat{u}$ in $\hat{H}_C$. Assume that $v = (a_1,a_2,\dotsc,a_i)$ and $u = (b_1,b_2,\dotsc,b_j)$. If $v' = (a_1,a_2,\dotsc,a_i,c_1,c_2,\dotsc,c_{i'}) \in V_C$ for some $c_1,c_2,\dotsc,c_{i'}$, and we have
\[
  f_C((a_1,a_2,\dotsc,a_i,c_1)) = \overline{C} \quad\text{and}\quad u' = (b_1,b_2,\dotsc,b_j,c_1,c_2,\dotsc,c_{i'}) \in V_C,
\]
then we define $f_{v,u}(v') = u'$. The idea here is that the subtrees of $H_C$ that have the nodes $v$ and $u$ as their roots and that are not contained in the subgraph~$\hat{H}_C$ (except for the root nodes) are isomorphic (up to a certain distance). The mapping $f_{v,u}$ is a partial isomorphism between such subtrees, as one can quite easily check. In what follows, we will use $f_{v,u}$ to show that the $r$-$\aSV$-bisimilarity of the nodes $((1,0,C))$ and $((2,1,C))$ in $\hat{H}_C$ can be extended to the supergraph~$H_C$.

For each $C \in \set{\cB,\cW}$, denote the nodes $\emptyset$, $((1,0,C))$ and $((2,1,C))$ of $H_C$ by $v_C$, $u_C$ and $w_C$, respectively. In accordance with our previously introduced notation, denote the corresponding nodes of the subgraph~$\hat{H}_C$ by $\hat{v}_C$, $\hat{u}_C$ and $\hat{w}_C$.

\begin{lemma}\label{lem:subtrees}
  Let $\hat{v},\hat{u} \in \hat{V}_C$ be grey nodes and let $t \in \N$ be such that $v \nbsmsv{t} u$. If $w \in \dom(f_{v,u})$, $\dist(w,v_C) < 2d-t$ and $\dist(f_{v,u}(w),v_C) < 2d-t$, then $w \nbsmsv{t} f_{v,u}(w)$.
\end{lemma}

\begin{proof}
  We proceed by induction on $t$. The base case $t = 0$ is straightforward: Since $\dist(w,v^C) < 2d$ and $\dist(f_{v,u}(w),v^C) < 2d$, we have $\deg(w) = \deg(f_{v,u}(w))$. Additionally, observe that we have $f_C(w) = f_C(f_{v,u}(w))$. It follows that we have $w \nbsmsv{0} f_{v,u}(w)$.

  For the inductive case, assume that the claim holds for $t = s$ and that $v \nbsmsv{s+1} u$. If $w = v$, then $f_{v,u}(w) = u$ and we have nothing to prove. Hence, assume $w \ne v$. Denote the neighbours of $w$ by $w_1,w_2,\dotsc,w_k$. Then the neighbours of $f_{v,u}(w)$ are $f_{v,u}(w_i)$, $i = 1,2,\dotsc,k$. We have $w_i \in \dom(f_{v,u})$ for all $i$. Additionally, since $\dist(w,v_C) < 2d-(s+1)$ and $\dist(f_{v,u}(w),v_C) < 2d-(s+1)$, we have $\dist(w_i,v_C) < 2d-s$ and $\dist(f_{v,u}(w_i),v_C) < 2d-s$ for all $i$. Now the inductive hypothesis implies that $w \nbsmsv{s} f_{v,u}(w)$ and $w_i \nbsmsv{s} f_{v,u}(w_i)$ for all $i$. Additionally, it follows immediately from the definition of $f_{v,u}$ that we have $\pi_C(w_i,w) = \pi_C(f_{v,u}(w_i),f_{v,u}(w))$ for all $i$. Now by Definition~\ref{def:nbisimsv} we have $w \nbsmsv{s+1} f_{v,u}(w)$. Hence the claim holds for $t = s+1$.
\end{proof}

\begin{lemma}\label{lem:bsmsubhtoh}
  Let $t \in \N$ and let $\hat{v},\hat{u} \in \hat{V}_C$ be such that $\dist(\hat{v},\hat{v}_C) < 2d-t$ and $\dist(\hat{u},\hat{v}_C) < 2d-t$. If $\hat{v} \nbsmsv{t} \hat{u}$, then $v \nbsmsv{t} u$.
\end{lemma}

\begin{proof}
  We prove the claim by induction on $t$. The base case $t = 0$ is easy: If $\hat{v} \nbsmsv{0} \hat{u}$, then $f_C(\hat{v}) = f_C(\hat{u})$, and thus $f_C(v) = f_C(u)$. As $v$ and $u$ are of the same colour and neither of them is a leaf node, $\deg(v) = \deg(u)$. Hence $v \nbsmsv{0} u$.

  For the inductive step, assume that the claim holds for $t = s$ and that $\hat{v} \nbsmsv{s+1} \hat{u}$, where $\dist(\hat{v},\hat{v}_C) < 2d-(s+1)$ and $\dist(\hat{u},\hat{v}_C) < 2d-(s+1)$. Denote the neighbours of $\hat{v}$ and $\hat{u}$ by $\hat{v}_1,\hat{v}_2,\dotsc,\hat{v}_d$ and $\hat{u}_1,\hat{u}_2,\dotsc,\hat{u}_d$, respectively. We have $\hat{v} \nbsmsv{s} \hat{u}$, and by definition, for each $\hat{v}_i$ there is $\hat{u}_{j_i}$ such that $\hat{v}_i \nbsmsv{s} \hat{u}_{j_i}$ and $\pi_C(\hat{v}_i,\hat{v}) = \pi_C(\hat{u}_{j_i},\hat{u})$, and vice versa. We have $\dist(\hat{v}_i,\hat{v}_C) < 2d-s$ and $\dist(\hat{u}_i,\hat{v}_C) < 2d-s$ for all $i$. Now the inductive hypothesis implies that $v \nbsmsv{s} u$, $v_i \nbsmsv{s} u_{j_i}$ for all $i$ and $v_{i_j} \nbsmsv{s} u_j$ for all~$j$.

  Since $v \nbsmsv{s} u$, nodes $v$ and $u$ are of the same colour. If they are of colour $C$, they do not have neighbours other than $v_1,v_2,\dotsc,v_d$ and $u_1,u_2,\dotsc,u_d$, respectively. Then it follows from the definition that $v \nbsmsv{s+1} u$. Otherwise, $v$ and $u$ are grey, and in addition to $v_i$ and $u_i$, $i = 1,2,\dotsc,d$, they have neighbours generated by rule~(H3) or rule~(H6). Denote those neighbours by $v'_1,v'_2,\dotsc,v'_{d-1}$ and $u'_1,u'_2,\dotsc,u'_{d-1}$, respectively, such that we have $f_{v,u}(v'_i) = u'_i$ for all $i$. Observe that $\dist(v'_i,v_C) < 2d-s$ and $\dist(u'_i,v_C) < 2d-s$ for all $i$. Now Lemma~\ref{lem:subtrees} shows that $v'_i \nbsmsv{s} u'_i$ for all $i$. In addition, the definition of $f_{v,u}$ implies that $\pi_C(v'_i,v) = \pi_C(u'_i,u)$ for all $i$. We have shown that conditions (B2) and (B3) hold also for the additional neighbours, and consequently $v \nbsmsv{s+1} u$. Hence the claim is true for $t = s+1$.
\end{proof}

Now we can combine our previous results to obtain bisimilarity between certain nodes in the graph~$H_C$ for each $C \in \set{\cB,\cW}$. Lemma~\ref{lem:pswtobisim} shows that $((1,0)) \nbsmsv{2d-3} ((2,1))$, where $((1,0))$ and $((2,1))$ are nodes in the graph~$G_d$. Observe that $g_C(\hat{u}_C) = ((1,0))$ and $g_C(\hat{w}_C) = ((2,1))$. Now Lemma~\ref{lem:bsmgtosubh} implies that $\hat{u}_C \nbsmsv{2d-3} \hat{w}_C$. We have $\dist(\hat{u}_C,\hat{v}_C) = 1 < 2d-(2d-3)$ and $\dist(\hat{w}_C,\hat{v}_C) = 1 < 2d-(2d-3)$. Hence it follows from Lemma~\ref{lem:bsmsubhtoh} that $u_C \nbsmsv{2d-3} w_C$, where $u_C$ and $w_C$ are neighbours of $v_C$ in the graph~$H_C$.

As in the proof of Theorem~\ref{thm:simulation}, we define a port numbering $p'_C$ for each $C \in \set{\cB,\cW}$ based on the generalised port numbering~$p_C$. Again, we need to preserve bisimilarity as well as have identical outgoing port numbers from nodes $u_C$ and $w_C$ towards node~$v_C$. Define function~$f$ from the set of all generalised ports of $H_C$ to $[2d-1]$ as follows: $f(1,\cB) = f(1,\cW) = 1$, $f(i,\cB) = 2i-1$ and $f(i,\cW) = 2i-2$ for all $i = 2,3,\dotsc,d$, $f(0,\cG) = 1$ and $f(i,\cG) = i$ for all $i = 1,2,\dotsc,d$. Then, if $p_C(v,a) = (u,b)$ for some nodes $v,u$ and port numbers $a,b$, set $p'_C(v,f(a)) = (u,f(b))$. Without too much effort, one can check that $p'_C$ is indeed a valid port numbering of $H_C$, and that we have $\pi'_C(u_C,v_C) = 1 = \pi'_C(w_C,v_C)$. Lemma~\ref{lem:pnchangebsm} implies that $(H_C,f_C,u_C,p'_C) \nbsmsv{2d-3} (H_C,f_C,w_C,p'_C)$.

To reach our ultimate goal, we need to define one more mapping. Define $h \colon V_\cB \to V_\cW$ as follows: if $v = (a_1,a_2,\dotsc,a_i) \in V_\cB$, where $i \ge 1$ and $a_1 = (b_1,b_2,C)$ for some $b_1 \ge 2$, set $h(v) = u$, where $u = (a_1,a_2,\dotsc,a_i) \in V_\cW$. Additionally, set $h(v_\cB) = v_\cW$. Thus, there is one subtree starting from a child of $v_\cB$, the one having the node~$u_\cB = ((1,0,\cB))$ as its root, that is excluded from the domain of $h$. Similarly, the subtree having $u_\cW = ((1,0,\cW))$ as its root is excluded from the range of $h$. See Figure~\ref{fig:bisim} for an illustration of the situation.
\begin{figure}
  \centering
  \includegraphics[page=7]{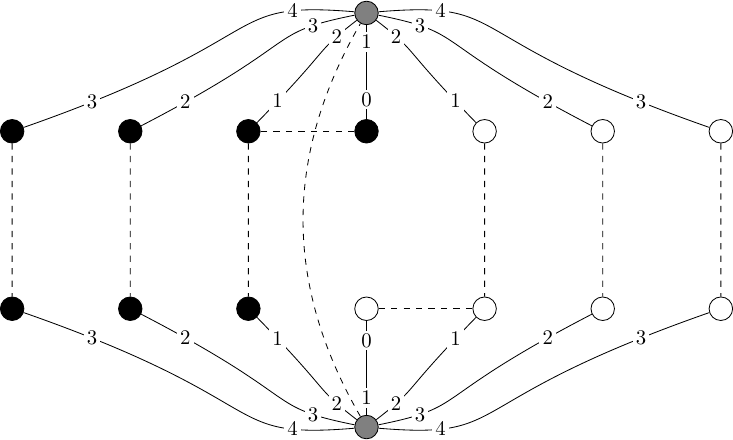}
  \caption{Graphs $H_{\cB,4}$ and $H_{\cW,4}$ up to distance one from the root nodes. The dashed lines represent $r$-$\aSV$-bisimilarity between nodes.}
  \label{fig:bisim}
\end{figure}

\begin{lemma}\label{lem:bsmhbhw}
  Let $v \in V_\cB$ and $u \in V_\cW$ be nodes such that $h(v) = u$. Then for all $t = 0,1,\dotsc,2d-2$ we have $(H_\cB,f_\cB,v,p'_\cB) \nbsmsv{t} (H_\cW,f_\cW,u,p'_\cW)$.
\end{lemma}

\begin{proof}
  We prove the claim by induction on $t$. The base case $t = 0$ is trivial: if $h(v) = u$, then by definition of $h$ we have $\deg(v) = \deg(u)$ and $f_\cB(v) = f_\cW(u)$ and therefore $v \nbsmsv{0} u$.

  For the inductive step, suppose that the claim holds for $t = s < 2d-2$. Consider two arbitrary nodes $v \in V_\cB$ and $u \in V_\cW$ such that $h(v) = u$. By the inductive hypothesis we have $v \nbsmsv{s} u$. If $v \ne v_\cB$, all the neighbours of $v$ are in the domain of $h$ and all the neighbours of $u$ are in the range of $h$. Furthermore, if $w$ is a neighbour of $v$, we have $\pi'_\cB(w,v) = \pi'_\cW(h(w),u)$, and by the inductive hypothesis, $w \nbsmsv{s} h(w)$. Now Definition~\ref{def:nbisimsv} implies that $v \nbsmsv{s+1} u$.

  If $v = v_\cB$, $v$ has one neighbour that is not in $\dom(h)$. That neighbour is $u_\cB = ((1,0,\cB))$. Similarly, $h(v) = v_\cW$ has one neighbour that is not in the range of $h$, namely $u_\cW = ((1,0,\cW))$. However, as shown above, we have $u_\cB \nbsmsv{2d-3} w_\cB$, and thus $u_\cB \nbsmsv{s} w_\cB$. Since we have also $w_\cB \nbsmsv{s} h(w_\cB)$, Lemma~\ref{lem:nbsmeqrel} implies that $u_\cB \nbsmsv{s} h(w_\cB)$. Additionally, we have
\[
  \pi'_\cB(u_\cB,v) = \pi'_\cB(w_\cB,v) = \pi'_\cB(h(w_\cB),u).
\]
Similarly, we have $u_\cW \nbsmsv{s} w_\cW$ and $w_\cW \nbsmsv{s} h^{-1}(w_\cW)$, from which we get $u_\cW \nbsmsv{s} \allowbreak h^{-1}(w_\cW)$. Additionally,
\[
  \pi'_\cW(u_\cW,u) = \pi'_\cW(w_\cW,u) = \pi'_\cW(h^{-1}(w_\cW),v).
\]
We have shown that conditions (B1)--(B3) hold even if considering also neighbours not handled by the mapping~$h$, and consequently we have $v \nbsmsv{s+1} u$. Thus the claim holds for $t = s+1$.
\end{proof}

Let $d \ge 2$ and $\Delta = 2d-1$. Then $H_{\cB,d},H_{\cW,d} \in \F(\Delta)$. Let $\aA \in \aSV$ be any algorithm with a running time at most $\Delta-1 = 2d-2$. Consider the execution of $\aA$ in the nodes $v_\cB \in V_{\cB,d}$ and $v_\cW \in V_{\cW,d}$. Now Lemma~\ref{lem:bsmhbhw} together with Lemma~\ref{lem:nbsmsvstate} implies that $\aA$ produces the same output in $v_\cB$ and $v_\cW$. Recall that for any valid solutions $S \in \Pi(H_{\cB,d},f_\cB)$ and $S' \in \Pi(H_{\cW,d},f_\cW)$ we have $S(v_\cB) \ne S'(v_\cW)$. Hence $\aA$ does not solve the problem~$\Pi$. This concludes the proof of Theorem~\ref{thm:parity}.

\begin{remark}
  Note that we could define a similar problem without local inputs, by encoding the colours in the structure of the graph. One way to do this is to add one new neighbour to each black node and two new neighbours to each white node. If $d \ge 3$, this does not increase the maximum degree of the graph. Then we could define the set of solutions to consist of, for example, mappings~$S$ such that $S(v) = 1$ if node~$v$ has an odd number of neighbours of an odd degree and $S(v) = 0$ otherwise. However, for illustrative purposes, it was beneficial the use a colouring instead.
\end{remark}

\section{Conclusions}\label{sec:concl}

To sum up, we have shown that the simulation technique used to prove $\SV = \MV$ is optimal in the following sense: breaking symmetry between incoming messages is always possible in time $2\Delta-1$, and there are graphs where $2\Delta-1$ rounds are strictly required. Furthermore, we have constructed a graph problem for which the difference in running time between algorithms in $\sSV$ and $\sMV$ is linear in $\Delta$. This settles the last significant open question related to the hierarchy studied by \textcite{hella15weak-models}.

\section*{Acknowledgements}

This manuscript is based on the author's master's thesis~\cite{lempiainen14msc} submitted to the Department of Mathematics and Statistics of the University of Helsinki. The author would like to thank Jukka Suomela for guidance and feedback as well as Lauri Hella and Juha Kontinen for comments on the thesis. A minor part of the research was conducted while the author was employed at the Department of Computer Science of the University of Helsinki.

\printbibliography

\end{document}